%% file: main.tex
\documentclass[11pt]{article}
\usepackage[letterpaper,margin=2.5cm]{geometry}

\usepackage{graphicx}				% Use pdf, png, jpg, or eps§ with pdflatex; use eps in DVI mode
\usepackage{url}
\usepackage{amssymb,amsmath,amsthm}
\usepackage{pdfpages}
\usepackage{float}
\usepackage{hyperref}
\usepackage{lipsum}
\usepackage{listings}
\usepackage{tikz}
\usepackage{times}
\usepackage{txfonts}
\usepackage{mdframed}
\newtheorem{theorem}{Theorem}[section]
\newtheorem{lemma}[theorem]{Lemma}

\newtheorem{definition}[theorem]{Definition}

\title{On the Local Communication Complexity\\ of Counting and Modular Arithmetic}

\author{

Bala Kalyanasundaram\\
Georgetown University\\
kalyan@cs.georgetown.edu

\and

Calvin Newport\\
Georgetown University\\
cnewport@cs.georgetown.edu

}

\date{}
\begin{document}

\maketitle

\begin{abstract}
In standard number-in-hand multi-party communication complexity, performance
is measured as the total number of bits transmitted globally in the network.
In this paper, we study a variation called local communication complexity
in which performance instead measures
the maximum number of bits sent or received  at any one player.
We focus on a simple model where $n$ players, each with one input bit,  execute a 
protocol by exchanging messages to compute a function on the $n$ input bits.  
We ask  what can and cannot be solved with a small local communication
complexity in this setting.
We begin by establishing a non-trivial lower bound
on the local complexity for a specific function
by proving that counting the number of $1$'s among the first $17$ input
bits distributed among the participants requires a local
complexity strictly greater than $1$.
We further investigate whether harder counting problems of this type
can yield stronger lower bounds,
providing a largely negative answer by showing that constant local
complexity is sufficient to count the number $1$ bits over the entire input, and therefore
compute any symmetric function. 
In addition to counting, we show that both sorting and searching can be computed in constant local complexity.
We then use the counting solution as a subroutine to demonstrate
that constant local complexity is also sufficient to compute many
standard modular arithmetic operations on two operands,
including: comparisons, addition, subtraction, multiplication, division, and exponentiation.  Finally we establish that function $GCD(x,y)$ where $x$ and $y$ are in the range $[1,n]$ has local complexity of $O(1)$.
Our work highlights both new techniques for proving lower bounds on this metric and the power of even a small amount of local communication.

\end{abstract}

\section{Introduction}
\label{sec:introduction}

In the standard study of number-in-hand multi-party communication complexity
the input bits for a function $f$ of size $n$ are partitioned among two or more
players.
The goal is for these players to work together to compute $f$ over the input bits
by transmitting information using network channels or a shared blackboard.
In the deterministic context, the {\em communication complexity} of a given protocol is the total number of bits transmitted
or recorded in the worst case over all possible inputs.

More recently, a natural variation was introduced that instead counts the maximum number
of bits sent or received by any one player~\cite{KM2015,boyle2018bottleneck}.
Though different names have been given to this measure, we call it {\em local (communication) complexity}.
In~\cite{KM2015}, the authors show 
connections between this metric and efficient distributed pattern recognition.
Later, in~\cite{boyle2018bottleneck},
the authors describe this metric as a ``fundamental area to explore,''
noting that measures of global communication obscure the local load at individual players,
a critical factor in settings where local processing is an important resource to conserve.
They further underscore this importance
by establishing formal theoretical connections between local complexity and multi-party secure
computation, streaming algorithms, and circuit complexity.

Informally, 
this model consists of $n$ players connected by network channels.
Each player gets a bit as input. 
They exchange messages according to some deterministic protocol to compute a function on these $n$ input bits.  
Each player maintains a single receive buffer into which all received bits are placed,
and they read from their buffers one bit at a time.
Executions proceed asynchronously to prevent the implicit encoding of information into silent rounds.

The study of local (communication) complexity remains in its early stages.
To date, for example, to the best of our knowledge
there are no known non-trivial lower bounds on the local
complexity of {\em specific} functions,
and only a small number of problems have been analyzed from the perspective of
identifying the number of bits that must be received locally.
The study of what can and cannot be solved with {\em small} local complexity
is further endorsed by the connection
%, first identified in~\cite{boyle2018bottleneck}, 
between this model and linear-size circuits.
Strong lower bound on the local complexity of a specific function would imply strong lower bounds on 
the circuit complexity of the function.\footnote{In~\cite{boyle2018bottleneck}, the authors
note that proving a given function $f$ requires local complexity in $\Omega(n)$ implies a circuit complexity
in $\Omega(n^2)$, which would represent a major breakthrough in the study of the latter field.}

{\bf Results:} We start by  tackling the open problem of producing a lower bound on local complexity for a specific function.
We focus on the natural challenge of counting the number of $1$'s among the input bits,
as this seems intuitively difficult to accomplish when restricted to communicating only a very
small number of bits at each player.
For a given bit sequence $a=a_1a_2...a_n$,
let $\#_1(a_i...a_j)$, for $1\leq i \leq j \leq n$, 
be the number of $1$ bits in this substring of $a$.
We formalize the {\em $(n,k)$-counting} function, denoted $g_k^n$, as follows:  $g_k^n(a_1a_2...a_n) = \#_1(a_1a_2...a_k).$  
Assume $n$ players numbered $1$ to $n$,
such that each player $i$ receives input bit $a_i$.
We say a protocol executed by these players
{\em solves} $(n,k)$-counting if at least one player outputs $g_k^n(a_1a_2...a_n)$,
and no player outputs something different.

In Section~\ref{sec:lower},
we prove every solution to $(n,17)$-counting has a local complexity 
strictly greater than $1$.
Though it is intuitive that you cannot count too high with such a small complexity,
we emphasize that {\em establishing} such a claim is less obvious; requiring, in this case, 
a novel combination
of indistinguishability and information theory techniques.

We next explore whether we can strengthen this lower bound by increasing $k$.
A natural conjecture, for example, is that $\Omega(\log{k})$ local complexity is required
to count the $1$'s in the first $k$ input bits.
In Section~\ref{sec:upper}, 
we disprove this conjecture
with a protocol that solves $(n,\frac{n}{10})$-counting with local complexity $2$.
We then show how to solve $(n,n)$-counting, and therefore solve every symmetric function,
with a local complexity of only $11$.
These solutions borrow techniques from circuit design to recursively
apply distributed
adder circuits to aggregate the sums in an efficient distributed manner.

We conclude this study of counting by considering the two related problems of sorting and search.
Counting the number of $1$ bits provides a straightforward solution to the problem
of sorting the input bits, as in a setting with binary input values, sorting reduces to arranging
all the $1$ bits to precede the $0$ bits. With this in mind, we show
how to transform a counting solution into a 
sorting solution at the cost of only one extra bit of local complexity.
Less obvious is the problem of searching for the position of the $k$th $1$-bit among the $n$ inputs.
Deploying a more involved strategy,
we show how to solve this problem with constant complexity using our counting solution as a subroutine.

Having established that symmetric functions can be computed with constant local complexity,
we next turn our attention to the important class of {\em 2-symmetric} functions.
Recall, a function $f$ is called 2-symmetric (or bi-symmetric) 
if the binary input bits can be split into two groups such that the function value does not change when we permute inputs within each of the groups.
Applying our previous counting strategy,
we can compute $x$ and $y$ with constant local complexity, in the sense that one player in the first partition learns its group's count is $x$,
and one player in the second partition learns its group's count is $y$.
The question remains whether
we can move forward from here to efficiently compute interesting 
functions of the form $f(x,y)$.

In Section~\ref{sec:arithmetic},
we provide some positive answers to this question
by describing strategies for computing $x \leq y$, $x+y$, $x-y$, $xy$, $x/y$, and $x^y$, all
modulo $n$,\footnote{Modularity is required as we are dealing with unary encoding
of inputs and outputs.}
 with only a constant local communication complexity.
These results underscore the surprising power of computation with low local complexity,
and the importance therefore of our lower bound.
They also provide useful efficient subroutines for other computations,
as many interesting problems have algebraic representations. 
To underscore this final point,
we show, perhaps surprisingly, that $GCD(x,y)$ where $x$ and $y$ and in the range $[1,n]$ can also be computed in constant local complexity.

\smallskip

{\bf Discussion:}
We emphasize that to the best of our knowledge,
this work is the first to establish non-trivial bounds
for what can and cannot be solved in a distributed fashion 
with small local complexity. 
There are many different problems we might have considered.
Our choice to study basic distributed counting and arithmetic tasks were motivated
by two factors: (1) they are natural and simple to define; and (2) they yield
sharp computability thresholds (e.g., what can be counted with $2$ versus $1$ bit
of local complexity). 

Our decision to focus on deterministic protocols is similarly motivated by the simplicity
of starting with the cleanest possible problem and model definitions.
We note that the power of randomness for the problems studied here is not obvious,
especially when considering the careful deterministic structuring
of communication patterns often deployed by constant local complexity solutions.
Non-determinism, by contrast, can be shown to be strictly stronger than determinism.\footnote{It is
straightforward to show how to easily compute every 2-symmetric function with constant local
complexity with a non-deterministic protocol. As mentioned in Section~\ref{sec:arithmetic},
however, a straightforward counting argument establishes that there exist 2-symmetric
functions with local complexities in $\omega(1)$ for deterministic protocols.}

%By allowing randomness, we introduce the concept of a probabilistic protocol and ask the question of a randomized protocol to solve the problem.   Does randomization give more power to the protocol if the protocol is $O(1)$?   While each node is capable of sending and receiving a constant number of bits, it can do a lot by choosing a random node to send a message.  On the other hand, allowing non-determinism one can easily compute every bi-symmetric function in $O(1)$ complexity.  Since not every bi-symmetric function can be computed in $O(1)$ by a deterministic protocol, we have a separation between deterministic and non-deterministic protocol.  

Of course, our work, combined with prior results~\cite{KM2015,boyle2018bottleneck},
still only scratches the surface when it comes to the deep 
exploration of local complexity.
Our goal here is not just to investigate this specific set of problems,
but to help instigate going forward the broader embrace of this intriguing and fundamental
metric by the distributed algorithm theory community. 

\section{Related Work}
The local communication metric studied here was introduced in~\cite{KM2015,boyle2018bottleneck}.
Our paper is perhaps best understood as a follow-up to~\cite{boyle2018bottleneck},
which motivated this model, but largely focused on problems with large local complexity,
leaving small local complexity as a topic for future exploration (a challenge we take on here).
Our formal model definition (Section~\ref{sec:model}) is somewhat more detailed than in~\cite{boyle2018bottleneck},
as such formality was needed to prove concrete lower bounds on specific functions.
Below we summarize existing work on communication complexity that predates and informs the work here
and in~\cite{KM2015,boyle2018bottleneck} on local complexity.

Naturally, local communication complexity
can be understood within the lineage of standard (global) communication complexity results, 
as it shares a commitment to minimizing the exact number of bits required for computing functions with inputs spread between players.
The study of (global) communication complexity started with Yao~\cite{Yao1979} in 1979.  
The main measure in this context is the total number of bits exchanged between two parties computing a function on their inputs.
Later, Chandra, Furst, and Lipton~\cite{CFL1983} introduced a {\em multi-party} communication complexity setting which is often referred to as a ``number on the head'' model:  there are now potentially more than $2$ players; each players knows all the other players' values, but not its own;
and they communicate by writing on a shared blackboard.
The complexity is the total number of bits displayed on the board.  Babai, Nisan and Szegedy~\cite{BNS1992}, among others,  
subsequently developed numerous bounds for this model (see the book by Kushilevitz and Nisan~\cite{KN1997} for an thorough review of this period).
Closer to our model is the subsequent work on so-called ``number in hand'' inputs, in which each player only knows its own input bits.
Numerous papers consider multi-party number~in-hand computation,
with several different communication assumptions: namely, the message passing,  blackboard, and coordinator models.  All of them measure the total number of bits sent/written in the network; e.g.,~\cite{CFL1983,BNS1992,Lotker,crama2011boolean,Drucker2014,Phillips2016,Yossef2004}.

Also relevant are {\em synchronous} models 
that similarly explore the amount of communication required to compute a function on data distributed among multiple
servers. 
In recent years, for example,
the massively parallel communication (MPC) model~\cite{beame2017communication}
has received increased attention. Inspired by Map Reduce/Hadoop-style systems (see~\cite{Jeffrey}), this model typically
bounds the amount of incoming communication at a given server in a given round by its local storage capacity.
The goal is to find good trade-offs between rounds and storage required to compute given classes of
functions (much of the early work in the MPC model, for example, focused on conjunctive queries on data~\cite{beame2017communication,ketsman2017worst}).
%
\iffalse
In~\cite{roughgarden2018shuffles}, Roughgarden et~al~consider a generalized MPC-style model 
that introduces an independent bound $s$ on the number of bits each server can receive per round.
The authors demonstrate that low-round computations in this model can be represented as low-degree polynomials. 
This connection allows them to apply existing results on approximating boolean functions with low-degree polynomials 
to derive lower bounds on the time complexity for computing these function in their model. 
\fi
%
Closely related to these models is the congested clique (e.g.,~\cite{Lotker}),
in which data is distributed among servers in a fully-connected network, 
and communication bounds are now placed on channels.

%As in~\cite{boyle2018bottleneck}, the authors deploy a circuit 
%simulation strategy to enable the hardness of circuit lower bounds to carry over to communication lower bounds in their model.
%In~\cite{chattopadhyay2017tight}, the authors use reductions from two-party communication problems
%(with known strong bounds) to obtain new lower bounds in the congested clique.

The above summary only samples the many papers that study the communication required 
for the synchronous computation of distributed functions.
Though related in spirit to our work on local communication complexity, 
the results do not directly apply to our setting.
In these synchronous models, the goal is to reduce the number of rounds required to compute a function,
whereas we minimize the exact number of bits sent or received at every player.

%While we define the local communication complexity as the maximum
%number of bits sent or received at any individual node during an execution, there are many minor variants of this measure.
%Our current choice seems to be the most appropriate one.
%Due to space limitations, we are unable to discuss the differences between various choices in this paper. 

%%%%%%%%%%%%%%%%%%%%%%%%%%%

\input{model}

\input{lower}

\input{upper}

\input{arithmetic}

%%%%%%%%%%%%%%%%%%%%%%%%%%%
\bibliography{refer}
%\bibliographystyle{plain}

%\newpage
%\appendix
%\input{appendix}

\end{document}

%% file: model.tex
\section{Model}
\label{sec:model}

Here we formalize our multi-party communication model and the local complexity metric we study in this setting.
Our definitions are more formal than in recent work on this metric~\cite{boyle2018bottleneck} as such specificity
is needed to study lower bounds for concrete functions. After introducing our formal definitions, we briefly discuss the specific
choices we made in attempting to nail down a model that balanced simplicity, tractability, and fidelity to existing work.

In more detail, we model a collection of $n>1$ deterministic computational processes (called both {\em players} and {\em nodes} in the following)
 executing in a variation of the standard asynchronous message passing model modified 
 to better suit the study of communication complexity.
We do not model messages arriving over discrete channels. Instead, the contents of 
 received messages are appended to a single string of received information that the receiver processes one bit at a time,
allowing for fine-grained control of exactly how many bits are consumed.
%We avoid the use of a distinct receive buffer for each incoming channel as this would leak information about
%the source of messages in a manner not captured in a count of actual bits received.
%Our model makes such information explicit. If you want to inform a receiver that you are the source of a given message,
%you must expend bits in your message to encode your identity. 

 \smallskip

{\bf Communication and Computation:}
Formally, let $V$ be the set of $n>1$ nodes in a fully-connected network topology.
We assume each node $u$ maintains a {\em receive string} $\psi_u$ which will store the bits of incoming messages.
This string is initialized to be empty.

For a given node $u$, if $\psi_u$ is not empty then the scheduler must eventually remove the first bit $b$ from $\psi_u$
and schedule a $recv(b)_u$ event at $u$.
As in the standard asynchronous message passing model,
when a $recv(b)_u$ event is scheduled,
node $u$ can update its state given the new bit $b$,
and send new messages to nodes in $V\setminus \{u\}$.
In more detail, if $u$ executes a $send(m,v)_u$ command during its processing of a $recv$ event,
for some message $m\in \{0,1\}^+$ and destination $v\in V\setminus \{u\}$,
then the bit string $m$ is appended to the end of $v$'s receive string $\psi_v$.
We treat the execution of the steps associated with a $recv(b)_v$ event,
including any $send$ commands, and the corresponding appending of sent message bits to other destinations' receive strings,
as one atomic step. Notably, this prevents bits in different messages from interleaving at a common receiver.

Also as in the standard asynchronous message model,
we assume that each node $u$ can also define an $init_u$ event,
which like a $recv(b)_u$ event can include $send$ commands.
For each $u\in V$, the scheduler must eventually schedule the $init_u$ event,
and it must do this before scheduling any $recv(b)_u$ events.
That is, each node gets a change to initialize itself before it starts processing incoming bits.

In this paper, we study algorithms that assume each player in the network is provided
a single bit as input.
For each node $u$ we use the notation $i_u$ to indicate $u$'s input bit.
Each node $u$ is also able to invoke $output(b)_u$ for $b\in \{0,1\}^+$, as part of its $recv$ and/or $init$ event step computation.
%If multiple nodes produce output then they must be the same.
A {\em problem} in this setting can therefore be understood as a mapping from each possible binary input assignment
to an integer $\{0,1\}^+$ in binary, which we can express as a function of the form $f:\{0,1\}^n \rightarrow \{0,1\}^+$.
%to the set of every correct binary output assignment for those inputs.
We say nodes in our model {\em solve} or {\em compute} such a function $f$ if provided input assignment $I\in \{0,1\}^n$,
at least one node outputs $f(I)$, and no node outputs anything different.

\smallskip

{\bf Local Communication Complexity:}
For a given execution $\alpha$,
and node $u\in V$, let $S_{\alpha}(u)$ and $R_{\alpha}(u)$ be the total number of bits sent and received by $u$, respectively, in $\alpha$.
We define the local communication complexity of a given $u\in V$ in $\alpha$, indicated $lcc_{\alpha}(u)$,
as follows: $lcc_{\alpha}(u) = \max\{ S_{\alpha}(u), R_{\alpha}(u) \}$.
We then define the local communication complexity of the entire execution $\alpha$, indicated $lcc(\alpha)$,
as $lcc(\alpha) = \max_{u\in V}\{ lcc_{\alpha}(u) \}$.
Let $P$ be a deterministic protocol.
We  define $lcc(P)$ to be the maximum $lcc(\alpha)$ value defined over every execution $\alpha$ of $P$.
Finally, for a given function $f:\{0,1\}^n \rightarrow \{0,1\}^+$, we define the local communication complexity of $f$,
also denoted $lcc(f)$,
to be the minimum $lcc(P)$ over every protocol $P$ that correctly computes $f$.

For the sake of concision,
we often use the slightly abbreviated phrase {\em local complexity} to refer to the local communication complexity
of a protocol or function.

%\subsection{Why this measure?}

\smallskip

{\bf Discussion:}
We opted for an asynchronous communication model as round numbers can leak information not captured by our complexity metric. 
%For example, in a synchronous setting,
%I can send you a value $i\in [x]$, for any $x$, using only $1$ bit by sending you this bit during round $i$.
We also avoided distinct channels for each sender/receiver pair as these channels provide for free
the identity of a given bit's sender. Because we focus in this paper on computing protocols with very small
local complexity, such leaks might end up significant.
Our solution to this issue was to introduce a common
receive buffer at each receiver on which incoming messages from all potential senders are appended.
In this setup, for example, if a sender wants to deliver a single bit to a receiver it can do so, and this bit will be appended to the receiver's buffer, but the receiver learns nothing about the source of the bit. If the sender wants the receiver to know who sent the bit, it has to actually send the up to $\log{n}$ bits required to encode its id.
%
%
%This would allow a sender to target any specific receiver and send a bit, say $1$.  Hence the sender sends only one bit, and the receiver does not know the sender.
%%On the other hand, if the sender wants to let the receiver know the address of the sender, then the sender must expend $\log n$ bits.
%This asymmetry is captured by the allowing only one receive buffer per node.  Without this, what you could compute is simply depend on only constant number of bits.
Another solution to avoid pairwise channels would have been
to deploy a central coordinator through which all bits are sent 
(an approach sometimes deployed in the existing
global communication complexity literature), but this centralization seemed incompatible with our focus on the local number of bits
sent and received at each individual player.

We also note that several similar definitions of our local complexity metric are possible.
We define local complexity at a given player $u$ in a given execution $\alpha$ as the max of the number of bits it sent ($S_{\alpha}(u)$) and the number
of bits it received ($R_{\alpha}(u)$).
One alternative would be to focus only on $S_{\alpha}(u)$---that is, the bits sent---when measuring local complexity.
This trivially allows, however, all functions to be computed with a minimum complexity of $1$ by having all players
send their input bit to a single pre-determined leader who locally computes the function.

Another alternative is to measure only $R_{\alpha}(u)$, the bits received.
This metric also seems to provide too much power to the players.
It is not hard to show, for example,
that it enables the computation of every bi-symmetric function with $O(1)$ local complexity.
The basic idea is to deploy the counting routines we present later in this paper that enables
one player in the first partition to learn the count $x$ of $1$ bits in the first partition,
and one player in the second partition to learn the count $y$ of $1$ bits in the second partition.
A close look at the routines reveal that the player that learns the count is dependent
on the count itself (roughly speaking, if the count is $i$, then the $i^{th}$ player in the partition
learns this fact).
The player that learns the count $x$ from the first partition
can now send a $1$ bit to every player in the second partition for which that player's corresponding
count would cause the bi-symmetric function to evaluate to $1$.

Giving these observations, our use of the maximum of $S_{\alpha}(u)$ and $R_{\alpha}(u)$ seemed  
the right choice to capture our intuitive understanding of local complexity, while avoiding
sweeping  solutions to large classes of problems.
More generally, we emphasize that there is rarely an obvious {\em best} way to model and measure multi-party communication complexity,
as evidenced by the variety of definitions in the existing literature. And as we have learned, all decisions
in such modelling evince  trade-offs.
We did our best here to arrive at a natural and straightforward definition that captures the local communication
we wish to study while sidestepping both trivializing assumptions and artificial difficulties.

%% file: lower.tex
\section{Counting Lower Bound}
\label{sec:lower}

A natural starting place to study what can and cannot be solved with small local communication complexity
is the fundamental task of counting.
In more detail, we study the local communication complexity of solving
 $(n,k)$-counting function.
Formally, we seek to identify a parameter $k$ for the 
counting function $g_k^n$ (defined in the introduction),
such that the local complexity of the function is strictly greater than $1$.
The core result of this section is a lower bound that establishes for any sufficiently large $n$,
$lcc(g_{17}^n) > 1$. 
We emphasize that is the first known lower bound on local complexity for a concrete function
(prior work~\cite{boyle2018bottleneck} contains only existential bounds based on counting arguments).

\subsection{Proof Summary}

At a high level, 
our proof strategy begins by focusing on local complexity
of the related {\em $(n,k)$-threshold detection} problem,
which requires the nodes to determine if at least $k$ of the input bits are $1$.
We prove that any protocol that solves this problem with local complexity
$1$  is highly
constrained in its operation, generating executions that can be understood as a
bit traveling in a chain, from one node to the next, with the final node making a decision.

Given such a structure,
we apply a combinatorial argument to argue that for a sufficiently large
constant threshold $k$, we
can construct two execution chain prefixes such that: (1) the correct output
is different for each chain (i.e., one chain has enough $1$'s to exceed the threshold
while the other does not); and (2) the node at the end of both prefixes sends
the same bit to the next link, obfuscating the actual contents of its predecessors.
The existence of these two prefixes can be deployed to generate an incorrect answer
in at least one of the two cases, contradicting the assumption that any algorithm
correctly solves threshold detection for this $k$ parameter.

Finally, once we bound threshold detection, we then use a reduction argument to obtain
our final bound for the more natural counting problem.

\subsection{Bounding Threshold Detection}
We begin by proving a lower bound on the local complexity of the
 {\em $(n,k)$-threshold detection} boolean function, 
that evaluates to $1$ if and only if at least $k$ out of $n$ input bits are $1$.
Formally, we use $f_k^n$ to indicate this function for a given pair of parameters $k$ and $n$,
and define it as:

\[f_k^n = \begin{cases}
1 & \text{if $\#_1(a_1...a_n) \geq k$,}\\
0 & \text{else.}
\end{cases}\]

Our goal is to prove that the following,
which establishes for sufficiently large $n$ value that threshold detection for $k=9$
requires a local complexity greater than $1$.

\begin{theorem}
\label{counting-lower}
Fix some network size $n>9$, threshold $k, 9 \le k \le n-8$.
It follows: $lcc(f_k^n) > 1$.
%Let $P$ be a protocol that computes $f_k^n$. It follows: $ccc(P) > 1$.
\end{theorem}

Before proceeding to main proof of this theorem, we establish some useful preliminaries
that formalize the constraints suffered by any threshold detection algorithm
with a minimum local complexity of $1$.
In the following,
we use the notation $P_1, P_2, \ldots P_n$ to represent the $n$ players.
We say that $P_i$ is an {\em initiator} with respect to a given input bit
if its initialization code for that input bit has it transmit a bit before receiving any bits.
A key property of a minimal local complexity environment is that a correct
protocol can only ever have one initiator:

\begin{lemma}
Fix some $n>4$,  $k$, $3<k<n$, and protocol ${\cal P}$ that computes $f_k^n$ with $lcc({\cal P}) = 1$.
There exists a player $P_i$ such that for every input assignment, 
$P_i$ is the only initiator among the $n$ players.
\label{lem:1bit:lem1}
\end{lemma}
\begin{proof}
We first argue that there must be at least one initiator.
Assume for contradiction that for some input assignment, $b_1,b_2, ..., b_n$, 
there are no initiators.
It follows that no players send or receive any bits.
Because we assume ${\cal P}$ correctly computes $f_k^n$, some player must output the correct answer without
ever having received any bits.
Let $P_i$ be a player that outputs. If it outputs $1$, meaning
there is at least $k$ input bits set to $1$ in our fixed assignment,
it will do the same even when set all other input bits to $0$---leading to an incorrect output.
Symmetrically, if $P_i$ outputs $0$,
it will do the same when we set all other input bits to $1$---leading to an incorrect output.
This contradicts the correctness of ${\cal P}$.

Moving forward, therefore,
we consider the case in which there are more than one initiator.
Once we have established that there cannot be more than one initiator, 
we will show that this one initiator must be the same for all input assignments.
Assume for contradiction that there exists some input assignment, $b_1,b_2, ..., b_n$
for which ${\cal P}$ has more than one initiator.
Let $P_i$ and $P_j$ be two such initiators.
Assume that the initialization code for $P_i$ with input bit $b_i$ has $P_i$ send bit $b_i'$ to player $P_x$,
and the initialization code for $P_j$ with $b_j$ has it send $b_j'$ to $P_y$.
Using these observations on the behavior of $P_i$ and $P_j$ we will identify an input assignment, $\hat b_1, \hat b_2, ..., \hat b_n$
that we can leverage to identify a contradiction.

Fix $\hat b_i = b_i$ and $\hat b_j = b_j$.
Fix $\hat b_x = \hat b_y = 1$.
Fix input values for the remaining players such that the total number
of $1$ bits in the assignment is exactly $k$ (because we assume $k \geq 4$,
this is always possible).

Consider an execution of ${\cal P}$ with assignment $\hat b_1, \hat b_2, ..., \hat b_n$.
Because this is an asynchronous systems an execution for a given input can depend
on the scheduling of send and receive events.
Assume a {\em round robin} scheduler that proceeds in {\em rounds} as follows:
During the first round, it visits each player in order $P_1, P_2$, and so on,
scheduling each player to complete its initialization transmission (if any).
In each subsequent round, it visits each player in order, for each, scheduling
the processing of bits transmitted in the previous round, and then completing
any new transmissions these received bits generate.
Call this execution $\alpha$.

In this execution we can break up communication into what we call {\em chains},
which capture the causal relationship of sends and receives beginning
with a given root player.
For example, if we fix $P_i$ as a root, and note that $P_i$ sends a bit to $P_x$,
which then enables $P_x$ to send a bit to some $P_{x'}$, and so on,
we note that there is a chain rooted at $P_i$ that begins $P_i, P_x, P_x',...$

Moving on, we note that by construction: $f_k^n(\hat b_1, \hat b_2, ..., \hat b_n) = 1$.
It follows that at least one player must output $1$ in $\alpha$.
Fix one such player $P_z$ that outputs $1$.
We argue that $P_z$ cannot be in both the chains rooted at $P_x$ and $P_y$. 
If this was the case, then at some point as we followed the chain from $P_x$ to $P_z$,
and the chain from $P_y$ to $P_z$,
some node $P^*$ would have to be visited in both.
This would require $P^*$ to receive at last $2$ bits which is not allowed in a protocol with
a local complexity of $1$

Without loss of generality, assume that $P_z$ is 
not in the chain rooted at $P_x$ (the other case is symmetric).
Consider the execution $\alpha_x$,
in which: (1) the input bit to $P_x$ is changed to $0$;
(2) we replace the round robin scheduler with one that first schedules the nodes in the communication
chain from $P_j$ to $P_y$ and on to $P_z$, in order, leading $P_z$ to output.
After this, we can revert to the round robin scheduler strategy to ensure all pending players get
chances to take steps.

By construction, $\alpha_x$ is indistinguishable from $\alpha$ with respect to $P_z$.
Therefore, $P_z$ will output the same value in $\alpha_x$ as $\alpha$.
Because we flipped the input value of $P_x$ in $\alpha_x$,
this output is wrong. 
This contradicts the assumption that ${\cal P}$ always correctly computes $f_k^n$.

We have now established that every input assignment has exactly one initiator.
We want to now show that this initiator is the {\em same} for every assignment.
To do so, assume for contradiction that  assignment $I_i$ has $P_i$ as its single initiator,
and assignment $I_j\neq I_i$ has player $P_j \neq P_i$ as its single initiator.
Consider a third assignment $I_{i,j}$ which is defined the same as $I_i$ with the exception
that player $P_j$ is given the same bit as in $I_j$.
We have now identified an assignment with two initiators.
We argued above, however, that every assignment has at most one initiator: a contradiction.
\end{proof}

The above lemma established that executions of protocols for threshold functions with minimum local complexity
have a single initiator, meaning they can be described as a sequence of player/message pairs.
We provide some notation to formalize this idea:

\begin{definition}
Fix a protocol with a single initiator $P_{i_1}$ and a local complexity of $1$.
We can describe an execution prefix of this protocol containing the first $j$ transmissions with a
single {\em chain} of the form:
$X =  (P_{i_1}, m_1), (P_{i_2},m_2), \ldots , (P_{i_j}, m_j),$
 where for each $x, 1 \leq x \leq j$, $P_{i_x}$ is the $x^{th}$ player to receive a bit,
and the bit it receives is $m_{x}$. If $x<j$, then $m_{x+1}$ describes the bit $P_{i_x}$ sent in response
to receiving $m_x$. Define $players(X) = \{P_{i_1}, P_{i_2}, \ldots , P_{i_j} \}$.
Because $P_{i_1}$ is an initiator, by convention we set $m_1 = \emptyset$.
We use the notation $X, (P_{\ell},m)$, for some $P_{\ell} \notin players(X)$, 
to indicate the concatenation of step $(P_{\ell}, m)$
to the end of chain $X$.
\end{definition}

When considering a chain $X=(P_{i_1}, m_1), (P_{i_2},m_2), \ldots , (P_{i_j}, m_j)$ that describes
an execution prefix, we can label each {\em step} $(P_{\ell}, m_j)$ in the chain
with a {\em value pair}, $(x_j,y_j)$,
where $y_j = j$ is the number of players involved in the chain up to and including $P_{\ell}$,
and $x_j$ is the number of these players with an input bit of $1$.
The value pair for a given step captures, in some sense, a possible information
scenario could generate that given step.

When considering the value pairs for a chain of an execution prefix
of a protocol computing a threshold function $f_n^k$, 
we say a pair of numbers $(x,y)$ is {\em valid} if two things are true:
it is {\em well-formed},
in the sense that the observed values could show up as a value pair
for a step in a chain (e.g., $x$ is not greater than $y$);
and they are {\em bivalent},
in that the values are compatible with both an output of $0$ or $1$
as the chain extends, depending on the details of the extension.
Formally:

\begin{definition}
Suppose the function under consideration is $f_k^n$.
We say that a pair $(x,y)$ is {\em valid} with respect to this function if
the values are:
\begin{enumerate}
    \item {\em Well-Formed}: $1 \le y < n$ and $x \le y$.
    \item {\em Bivalent}: $0 \le x < k$ and $(k-x) \le (n-y)$.
\end{enumerate}

\end{definition}

When considering chains for an execution of a protocol that computes a given $f_k^n$ with a local
complexity of $1$, we might want to ask the question of what are the properties of
input bit assignments could possibly lead to a given step $(P_i, m)$.
We formalize this question by defining a set $D$ that captures all value pairs compatible
with a given step:

\begin{definition}
\noindent
Given a protocol ${\cal P}$ that computes a function $f_k^n$ with local complexity $1$,
a player $P_i$, $1 \leq i \leq n$,
and bit $b\in \{0,1\}$,
we define the set $D(P_i,b)$ to contain every pair $(x,y)$ that satisfies
the following properties;
\begin{enumerate}
    \item $(x,y)$ is valid with respect to $f_k^n$, and
    \item $\exists$ input assignment $I$ for ${\cal P}$ that induces a chain that includes
a step $(P_i, b)$ labeled with value pair $(x,y)$.
\end{enumerate}

\end{definition}

\iffalse

\begin{figure}
\includegraphics[width=1.0\textwidth]{threshold-lb.pdf}
\label{fig1}
\end{figure}

See Figure~\ref{fig1} for an example of $D(P_i,m)$ for a protocol for computing the threshold function $f_2^4$.   Finally, it is trivial to see the following fact that captures a simple property of what happens after we reach a state $(P_i,m)$ where a player $P_i$ receives a message $m$.  In simple terms, any path taken by the protocol after this state is the same for all paths that led to this state.  The follows from $lcc=1$ and the only information a player has about the past 
is the message.
\fi

%In the proof of Theorem ~\ref{counting-lower} below we will argue that every vaild $(x,y)$ will appear in at least one $D(P_i,m)$ and $|D(P_i,m)| \le 3$.   By pigeon-hole principle this is not possible.

Before tackling our main theorem,
we have one last useful result to establish:
that every valid pair for a given $k$ and $n$ value can show up in some chain.
%The proof for the following claim can be found in the appendix.

\begin{lemma}
Fix a protocol ${\cal P}$ that computes a function $f_k^n$ with local complexity $1$.
Let $(x,y)$ be any valid pair for $f_k^n$.
There exists a player $P_i$ and bit $b$,
such that the set $D(P_i, b)$, defined with respect to ${\cal P}$,
includes $(x,y)$.
\label{lem:lbit:2}
\end{lemma}
\begin{proof}
Fix a ${\cal P}$, $k$, $n$, and $(x,y)$ as specified by the lemma statement.
By Lemma~\ref{lem:1bit:lem1}, every execution of ${\cal P}$ has a single initiator
and can be described by a chain.
We will create such a  chain step by step, setting
the input bit for each player in the chain only after they appear in the chain receiving their bit.
In more detail, for the first $y-x$ players that show up in the chain, we set their input to $0$.
For the remaining $x$ players, we set their input bits to $1$.
Notice that we can set these input bits after a player shows up in the chain,
because in a setting with local complexity $1$, after a player receives a bit,
if it cannot output, it must send a bit to keep the execution going, regardless
of its input. Its input bit {\em can} determine which player receives its transmission,
which is why we have to build this assignment dynamically as the chain extends.

A straightforward contradiction argument establishes that none of the first $y-1$
players in this chain can avoid transmitting, and therefore extending the chain.
This follows because, as constructed, this chain remains
bivalent until at least player $y$, in the sense that at every step, there exists an assignment
of input bits to the players that have not yet participated that makes $0$ the correct output,
and an assignment that makes $1$ the correct output.

Let $(P_i, b)$ be step $y$ in this chain.
By construction: $(x,y)\in D(P_i, b)$.
\end{proof}

\iffalse
%\label{lem:lbit:2}
\begin{proof}
Fix a ${\cal P}$, $k$, $n$, and $(x,y)$ as specified by the lemma statement.
By Lemma~\ref{lem:1bit:lem1}, every execution of ${\cal P}$ has a single initiator
and can be described by a chain.
We will create such a  chain step by step, setting
the input bit for each player in the chain only after they appear in the chain receiving their bit.
In more detail, for the first $y-x$ players that show up in the chain, we set their input to $0$.
For the remaining $x$ players, we set their input bits to $1$.
Notice that we can set these input bits after a player shows up in the chain,
because in a setting with local complexity $1$, after a player receives a bit,
if it cannot output, it must send a bit to keep the execution going, regardless
of its input. Its input bit {\em can} determine which player receives its transmission,
which is why we have to build this assignment dynamically as the chain extends.

A straightforward contradiction argument establishes that none of the first $y-1$
players in this chain can avoid transmitting, and therefore extending the chain.
This follows because, as constructed, this chain remains
bivalent until at least player $y$, in the sense that at every step, there exists an assignment
of input bits to the players that have not yet participated that makes $0$ the correct output,
and an assignment that makes $1$ the correct output.

Let $(P_i, b)$ be step $y$ in this chain.
By construction: $(x,y)\in D(P_i, b)$.
\end{proof}
\fi

We now have all the pieces required to tackle the proof
of Theorem~\ref{counting-lower} by deploying a novel combinatorial argument.
We begin by fixing $k=9$ and $n=17$.  
We show that every valid $(x,y)$ must show up in at least one $D(P_i,b)$ set.
Because there are fewer such sets than valid pairs,
the pigeonhole principle tells us that some $D(P_i,b)$ must have multiple
pairs. (It is here that the specific values of $k$ and $n$ matter, as they dictate the number
of possible valid pairs.)

At a high-level, that means when $P_i$ receives bit $b$ in a chain,
there are multiple possibilities regarding how many one bits appear
in the chain leading up to this step. Because $P_i$ cannot distinguish
between these value pairs we can, with care, craft
an execution extension in which the protocol outputs the wrong value.
In making this argument, extra mechanisms are required to deal with the possibility 
that the first player in the chain ends up the last player as well (this is possible
because an initiator begins an execution without having yet received a bit).
See Appendix for the proof of Theorem ~\ref{counting-lower}.

Once we have established our impossibility for $k=9$ and $n=17$,
we apply a reduction argument to generalize the results for larger values,
by showing such solutions could be used to solve our original fixed-value case.
This argument leverages the ability of the $n=17$ players to locally simulate
additional players without expending extra communication bits.

\begin{proof}(of Theorem ~\ref{counting-lower})
Assume for contradiction that there exists a protocol ${\cal P}$ that
computes the $(17,9)$-threshold detection function, denoted $f_9^{17}$,
with a local complexity of $1$.
We will prove that this protocol must sometimes output the wrong answer,
contradicting the assumption that its correct.
We will then generalize this argument to larger $k$ and $n$ values using a reduction argument.

%By Lemma~\ref{lem:1bit:lem1}, ${\cal P}$ has a single initiator.
%By convention, we label this player $P_1$.
%Every chain describing an execution of ${\cal P}$ begins with the step $(P_1, \emptyset)$.
%We also note that a system with $n=17$ and $k=9$, the set $V$ of valid value pairs that could
%label steps of chains can be defined as follows: $VP_9^{17} = \{ (x,y):  0\le x \le 8, 1\le y \le 16, 0 \le (y-x)\le 8 \}$. Simple counting establishes $|VP_9^{17}| = 80$.

Let $V_9^{17} = \{ (x,y):  0\le x \le 8, 1\le y \le 16, 0 \le (y-x)\le 8 \}$ be the set of
valid value pairs for $f_9^{17}$.
Simple counting establishes that $|V_9^{17}| = 80$.
By Lemma~\ref{lem:lbit:2}, every $(x,y)\in V_9^{17}$
must show up in some pair set $D(P_i,b)$.
Because there are $n=17$ possible players and $2$ possible bits,
there are $34$ possible pair sets.
The pigeonhole principle therefore establishes
that there exists a player $P_i$ and bit $b$
such that $|D(P_i,b)| \geq \lceil 80/34 \rceil = 3$.

Going forward,
we will  use this {\em target} pair $(P_i, b)$ to create our contradiction.
Consider the values in $D(P_i, b)$.
By the definition of $D$, each $(x,y)$ in this set is associated with at least one 
chain
that ends with step $(P_i, b)$, includes $y$ players, exactly $x$ of which have input bit $1$.
Call these {\em source} chains.
Label each pair in $D(P_i,b)$ with one of its source chains.
Further label each of these source chains with a compatible
input value assignment for the players in the chain (i.e., what
is the input assignment to these players that generates the chain;
choosing one arbitrarily if more than one assignment would create the same chain).

Because there are at least three pairs in $D(P_i,b)$,
there must be two such pairs, $(x_1, y_1)$, $(x_2, y_2)$,
such the initiators in their respective source chains
must have the same input bit $b_1$ in their compatible value assignment.
Notice, by Lemma~\ref{lem:1bit:lem1}, each of these source chains start 
with the same initiator.
To simplify notion, let us call this initiator $P_1$ for the purposes of our proof.
As will become clear, it is important that $P_1$ has the same input bit in both source
chains as it is possible that eventually full chain we consider will loop back to $P_1$.

Moving forward, we will use $X = (P_1,\emptyset), \ldots, (P_a,m_a), (P_i,b)$
to reference the relevant source chain associated with $(x_1, y_1)$,
and $I_1$ to be the relevant compatible input assignment.
We define $Y$ and $I_2$ analogously but now with respect to $(x_2,y_2)$.
Recall that by construction $P_1$ is assigned the same bit $b_1$ in $I_1$ and $I_2$.

We consider two cases concerning the players that shop up in chains $X$ and $Y$:

\bigskip

\noindent {\bf Case 1:} $players(X) \neq players(Y)$.

     By definition: both $(x_1,y_1)$ and $(x_2,y_2)$ are valid value pairs. 
     It follows that both are bivalent, meaning that the input bits of the players
     that have sent or receive bits so far are not sufficient to determine the value
     of the function. A straightforward contradiction argument establishes that no player
     in either chain can output until the chain extends further, 
     as if any player outputs $0$, the bits of the players
     not in the chain can be set to $!$ to make that answer incorrect, and if any player
     outputs $1$, the remaining bits can be set to $0$.
     Therefore, when we get to step $(P_i,b)$ in both chains, the output has not yet been determined.

     Because  $players(X) \neq players(Y)$, there must be a player $P_j$ in one set but not the other.
     Without loss of generality, say $P_j$ is only in $players(X)$.
     Fix any possible extension $Y,Z$ of $Y$ (where ``possible" means there is an input assignment to
     the players in $Z$ such that when combined with the fixed assignments for players in $Y$,
     $Y,Z$ describes the steps of the resulting execution).
     
     The key observation is that it must be the case that $P_j \notin players(Z)$.
     This follows because if $Z$ extends $Y$ then it also extends $X$,
     as both $X$ and $Y$ end with the same step: $P_i$ receiving $b$.
     However, $X,Y$ cannot occur because it features $P_j$ both in $X$
     and $Z$, meaning that this chain would require the same non-initiator player\footnote{We know
     that $P_j \neq P_1$ because it only shows up in on o the two chains, $X$ and $Y$, whereas
     $P_1$ is the single initiator in both.}
     $P_j$ to receive $2$ bits, which it cannot given our assumption of a local complexity of $1$.
     
     This observation creates an obstacle for the correctness of our protocol.
     We have just established that every way we can extend $Y$ must omit $P_j$.
     Consider the extension $Z$ that occurs when we fix the input bits of all
     players that are not in $players(X)$ and not $P_j$, such that the total number
     of $1$ bits is $k-1$. The execution corresponding to $Y,Z$ must eventually output.
     It does so, however, without $P_j$ sending a bit.
     If this execution outputs $1$, then it is incorrect in the case where $P_j$ has bit $0$,
     and if it outputs $0$, then it is incorrect in the case where $P_j$ has bit $1$.

\bigskip
\noindent {\bf Case 2:} $players(X) = players(Y)$:

If $players(X) = players(Y)$ then it follows that $y_1 = y_2$.
Because $(x_1,y_1) \neq (x_2, y_2)$, it also follows that $x_1 \neq x_2$.
That is, the number of $1$ bits encountered before $P_i$ receives $b$ is
different in $X$ versus $Y$. Player $P_i$, of course, receives the same bit
in both cases, so it must proceed without knowing if the count is $x_1$ or $x_2$.
The only player in these chains that can possibly receive another bit
is the common initiator $P_1$,
as only initiators send a bit before receiving any bit.
Since this initiator has the same input bit
in both $X$ and $Y$ (here is why it was important that we earlier identified
two chains that satisfied this property), our protocol must eventually output
without ever learning the true count of $1$ bits in the prefix leading up to $P_i$'s step.

To formalize this intuitive trouble, assume without loss of generality that
$x_1 > x_2$. Because $(x_1,y_1)$ is valid, we know $x_1 < k$.
Consider the extension $X,Z$ that occurs when we set exactly $k-x_1$ of the
players {\em outside} $players(X)$ to have input bit $1$.
The input assignment corresponding to $X,Z$ includes exactly $k$ $1$ bits,
therefore some step in $Z$ must correspond to a player outputting $1$.

If we consider this {\em same} input assignment for the players outside of $players(X) = players(Y)$,
we will get the {\em same} extension $Y,Z$, as the last step in $Y$ is the same as the last step in $Z$.
The set $players(Z)$ is disjoint from the set $players(Y)$ with the possible exception of $P_1$,
as it is possible that the initiator ends the chain it started.
By definition, however, all players in $players(Z)$ have the same input bit in 
the assignments corresponding to $X,Z$ and $Y,Z$ (recall, we selected $X$ and $Y$ specifically
because their corresponding assignments give $P_1$ the same bit),
and they receive and send the same bits in both, so the player in $Z$ that outputs $1$ in 
the execution corresponding to $X,Z$ also outputs $1$ in the execution corresponding to $Y,Z$.
This latter output, however, is incorrect, as the number of $1$ bits n the corresponding
input assignment is strictly less than $k$.

\bigskip

\noindent We have just established that any fixed protocol attempts to compute $f_9^{17}$ with local
complexity $1$ can be induced to output the wrong answer. This contradicts our assumption
that such a protocol exists.
We now use this result the generalize our impossibility to larger $k$ and $n$ values.

Fix any $k$ and $n$ values where $9 \le k \le n-8$, as specified by the theorem.
Assume for contradiction we have a protocol ${\cal P}$ that computes $f_k^n$ for these values with local complexity $1$.
We will now define a protocol ${\cal P'}$, defined for $17$ players,
that simulates ${\cal P}$ in a distributed
fashion to compute $f_9^{17}$ with a local complexity of $1$---contradicting our above result that no such protocol exists.

In more detail, protocol ${\cal P'}$ has the players in $P_{real} = \{P_1,\ldots, P_{17}\}$
collectively simulate the players in $P_{sim} = \{P_{18},\ldots P_n\}$,
such that first $k-9$ players in $P_{sim}$ start with input bit $1$, and the rest (if any remain) with input bit $0$.
Our assumption that $k \leq n-8$ ensures that there are at least $k-9$ players in $P_{sim}$ to initialize with a $1$ bit
(as $k \leq n - 8$ implies that $k - 9 \leq n - 17 = |P_{sim}|$, as needed).
Notice, the output in this simulated setup is $1$ if and only if at least $9$ of the players in $P_{real}$ have input but $1$.
Therefore, if we can correctly simulate ${\cal P}$ in this setting we can compute $f_9^{17}$.

We are left then to show how to correctly implement this simulation.
We can assume without loss of generality that the single initiator in $P_1$ (as established
by Lemma~\ref{lem:1bit:lem1}).
It begins by running ${\cal P}$ as specified. If the protocol has it
send
a bit to a player in $P_{real}$,
then it can send the bit as specified.
If it is instead instructed to send a bit to a player in $P_{sim}$,
it simulates locally that player receiving the bit and simulates that player's subsequent
send. It continues this simulation until a bit is sent to a player in $P_{real}$,
at which point the bit is actually sent to that player by $P_1$. Continuing in this manner, 
${\cal P'}$ can simulate ${\cal P}$ running on all $n$ players.

Two properties support the correctness of this simulation. First, each
player in $P_{sim}$ can receive at most one message, so each player only needs 
to be simulated once, eliminating the need for multiple players in $P_{real}$
to coordinate the simulation of a single $P_{sim}$ player. 
Second, given a chain that starts with a player $P_i\in P_{real}$, moves through one or more players
in $P_{sim}$, and then ends at a player in $P_j\in P_{real}$, it is valid for $P_i$ to send a bit
directly to $P_j$, as $P_i$ saved the bit it was instructed to send to a $P_{sim}$ player by ${\cal P}$ (as
it just locally simulated this communication), and the local complexity model does not convey
the source of a received bit, so $P_j$ cannot distinguish from which player an incoming bit was sent.
 \end{proof}

\subsection{ Generalizing from Threshold Detection to Counting}

We now leverage our result on threshold detection to derive a 
lower bound on any protocol that solves counting.
The reduction here is similar in construction to the argument deployed
in the preceding proof to generalize the $(17,9)$-threshold detection result to larger values of $n$.

 \begin{theorem}
\label{counting-lower2}
For every $17 \leq k \leq n$,
it follows: $lcc(g_k^n) > 1$.
%every protocol that solves $(n,k)$-counting has a local complexity greater than $1$.
%Let $P$ be a protocol that computes $f_k^n$. It follows: $ccc(P) > 1$.
\end{theorem}
% counting-lower2
\begin{proof}
Assume for contradiction that there exists a protocol ${\cal A}$
that solves $(n,k)$-counting with local complexity $1$ for some $17 \leq k \leq n$.
We can use ${\cal A}$ to define a new protocol ${\cal A'}$
that solves $(k,k)$-counting also with local complexity $1$.
To do so, we deploy the same strategy from the reduction argument
deployed in the proof of Theorem~\ref{counting-lower},
and have the $k$ players participating in protocol ${\cal A'}$
execute ${\cal A}$, locally simulating the $n-k$ {\em extra} players expected by ${\cal A}$.
They can simulate these extra players all starting with input bit $0$.

We now have a protocol ${\cal A'}$ that solves $(n', n')$-counting for some $n' \geq 17$.
We can use ${\cal A'}$ to compute the $(n',9)$-threshold detection function
in a network of size $n'$:
run ${\cal A'}$; if ${\cal A'}$ has one of the first $8$ players output $1$,
then that same player outputs $0$ for the threshold detection result; otherwise,
if a player beyond position $8$ outputs $1$ in ${\cal A'}$, that same player
outputs $1$ for the threshold detection result.

By Theorem~\ref{counting-lower}, however,
$(n',9)$-threshold detection cannot be computed for $n' \geq 17$ with local complexity $1$:
a contradiction.
\end{proof}

%% file: upper.tex
\section{Counting Upper Bounds}
\label{sec:upper}

In the previous section, we proved that you cannot count to $17$ with only a single bit
of local communication complexity.
Here we explore how much additional complexity is required to count to higher values.
%We prove that as we grow $k$,
%the local communication complexity required to 
%solve $(n,k)$-counting (compute $g_k^n$)
%does {\em not} grow with $k$.
%It turns out instead that a constant local complexity is always sufficient.
We divide this investigation into three questions: 
(1) what is the largest $k$ such that we can solve $(n,k)$-counting with a local complexity of $2$?;
 (2) what local complexity is required to solve $(n,n)$-counting?; and (3) what other problems
 can be easily solved with low local complexity using these counting strategies as a subroutine?

We tackle the second question first, describing how to solve $(n,n)$-counting
with constant local complexity.
This disproves the reasonable conjecture that the local complexity of $(n,k)$-counting must grow
as a function of $k$ (e.g.,  $\log{k}$).
We  then turn our attention to the question of how high we can count with
a local complexity of only $2$.
Our solution,
which deploys  ideas from our $(n,n)$-counting protocol in a more complex construction,
solves $(n,(n/10))$-counting, demonstrating a stark discontinuity between $1$ and $2$
bits of local complexity.
Finally, we establish two corollaries that deploy these strategies
 to solve both sorting and search with constant complexity.

\subsection{Solving (n,n)-Counting with Constant Local Complexity}
We begin by considering $(n,n)$-counting, which we prove can be solved with local complexity $11$.
As mentioned, this disproves the natural conjecture that the local complexity of $(n,k)$-counting must 
grow with $k$. For ease of presentation, we begin with a strategy that assumes $n$ is a power of $2$.
This result can be generalized to an
arbitrary $n$ at the cost of a more involved protocol.

We formalize this result below in Theorem~\ref{thm:sym}.
Its proof depends on the construction of a counting protocol
that carefully minimizes the number of bits each individual node sends or receives.
Given the importance of this strategy to all the results that follow in this
section,
we begin with a high-level summary of our protocol before proceeding with
its formal description and analysis in the proof of Theorem~\ref{thm:sym}.

\bigskip

{\bf Protocol Summary:}
At a high-level, the protocol
that establishes Theorem~\ref{thm:sym}
operates in two phases.
During the first phase, a count of the number of $1$ bits is aggregated into
a distributed counter in which $\log{n}$ nodes each hold a single counter bit.
In slightly more detail, we start by partitioning the nodes into $\Omega(n)$
groups of constant size, and for each group aggregating the count of their $1$
bits into a distributed counter of constant size.
We then begin repeatedly pairing up counters and having them sum up their
values in a distributed manner using strategies derived from arithmetic circuit design, 
allowing them to calculate a sum without any single node involved
in these counters needing to send or receive more than a constant number
of bits.

At the end of the first phase,
we have aggregated the total count into a distributed counter of size $\log{n}$.
In the second phase,
the nodes that hold the counter bits
help direct a descent through a binary tree with one leaf for each possible count.
The goal is to arrive at the leaf corresponding to the value stored in the counter,
consolidating knowledge of the entire count at a single node.
To do so, each bit of the counter informs the nodes implementing its corresponding level
of the tree its counter bit value,
propagating it in a chain of transmissions to prevent too much local communication.
Therefore, when the tree descent arrives at each level,
the specific node at which it arrives knows which sub-tree on which to advance the descent.

The proof that follows  details each of the steps that makes up these phases,
carefully accounting for the exact number of bits sent and received
in their implementation. 

\bigskip

{\bf Formal Result:}
We now show that when implemented and analyzed carefully, the local complexity of the  protocol
summarized above is no more than $11$.

\begin{theorem}
For every $i \geq 1$ and $n = 2^i$ there exists a protocol that solves $(n,n)$-counting with a local communication complexity of 
$11$. This protocol can be used to compute any symmetric boolean function with the same local complexity.
%
%Fix any symmetric boolean function $f:\{0,1\}^n \rightarrow \{0,1\}$.
%It follows: $lcc(f) \leq 11$.
%There exists a protocol $P_f$ that computes $f$ with $ccc(P_f) = 11$.  
\label{thm:sym}
\end{theorem}
\label{thm:counting}
\begin{proof}
%Assume that $n$ is a power of 2. 

We describe the repeated binary addition bottom-up process to store the count in binary.  
Now imagine a complete binary tree $T$ with $n/16$ groups as leaves.  
The protocol proceed one level at a time, starting from the leaf level, until it reaches the root.  
At each level, the protocol maintains the number of $1$'s in the sub-tree rooted at that level.

At the leaf level, we group 16 leaves at a time.  There are $n/16$ leaves. 
Starting from the first group of 16, run a simple and naive count protocol to count the number of $1$'s and 
once the count is complete a message is sent to the first member of the next leaf to start the process. 
Within a group of 16, run a four-bit protocol from first member of the group to the 16th/last member of the group to count the number of $1$'s in a linear chain fashion. 
At the end, the sum is represented as 5 bits. The last member retains the least significant bit of the sum and 
sends one bit each to members 12, 13, 14, 15 such that these bit values put together form the sum in binary.  
As explained before, the 16th member of the group then sends a bit to the first member of the next leaf to start the counting process.   
In the end, the last leaf sends a bit to start the addition process at the next level of the tree $T$.  
The recipient of the message is predetermined and will become clear when the processing of the next level is explained.   
At the end of the leaf-level processing, each member sends and receives at most 5 bits each. 

We now describe a bottom-up counting protocol that computes the sum of these counts in binary using a simple addition with carry and 
store the results of intermediate sums as binary.  The bits of the resultant binary number are stored distributively where every 
member stores a bit of the binary sum.   This significantly reduce the  local complexity.  
In order to show this, we keep track of the number of members available and the number of members used thus far in the process.   
Suppose we have computed the sum of $1$s in groups of size $2^k$.  We will show how to compute the sum for a group of size $2^{k+1}$.
In each group of size $2^k$, the binary bits of the number of $1$s are kept in $k+1$ distinct locations.  Let $U(k)$ be the number of
locations used exactly once for a group of size $2^k$.  Since exactly $k+1$ new locations are needed for the group of size $2^k$, 
the recurrence relation for $U(k)$ is $U(k)= 2 U(k-1) + k+1$ and $U(4) = 5$.  Solving this recurrence relation, we get $U(k) = 3~2^{k-2} - (k+3)$.
Since $U(k) \le 2^k$, for all $k \ge 4$, each location stores the sum at most once and the locations can be fully pre-specified for each iteration.
Each member has full knowledge of this participation.

In the computation of the sum for a group of size $2^k$, two sub-groups of size $2^{k-1}$ each has the sum stored in $k$ locations each.  
There are two phases in this computation.  In the first phase, called deposit phase, the bits to be added and deposited into a new location each. 
In the second phase, called carry-add phase, the bits are added and the carry is rippled and the carry information, the third bit, 
is a signal to perform the computation.     

The deposit phase for group of size $2^k$ begins after the completion of carry-add phase for all sub-groups of size $2^{k-1}$ and it is initiated when a bit is received from the player who completes the carry-add 
phase for the last group of size $2^{k-1}$.  In the deposit phase, the two least (respectively ith-least) significant bit locations 
from previous computations send their bits to the least (respectively ith-least) significant bit location for the resultant sum for the group of size $2^k$.

Once this depositing process is complete, the last member will send a message "0" to the member representing least significant bit of the new sum to start the carry-add phase. When a member of the new sum has received three bits, it computes the sum bit and the carry bit.  
It stores the sum bit and sends the carry bit to the next location.  
Recall that there are many groups of size $2^k$ exist and all of them must be calculated before we move on to groups of size $2^{k+1}$.

While carry information represents the third bit which triggers the calculation, the most significant bit calculation a group of size 
$2^k$ sends a 0-bit to the least significant bit for the next 
group to continue the calculation.  
The most significant bit calculation of the last group of size $2^k$, sends a message to start the deposit phase for group of size $2^{k+1}$.  
It is important to note that who participates in what is fully determined beforehand and everyone has full knowledge of this information.

We now calculate the number of bits sent and received by any node during this process.  In the base case, each node sends and receives at most 5 bits each.  
Since each node participates in the calculation of one sum, it receives 3 bits (two bits plus a carry or control bit), 
and sends two bits (resultant sum bit to the next group-size and a carry bit to the current group-size).

Note that final sum is in the range $[0,n]$ where both $0$ and $n$ are included.  It occupies $1+ \log n$ bits and the most significant bit is "1" if and only if the sum is $n$.  If this is the case, then the member corresponding to the most significant bit of the sum can declare the output of function.  For now, let us assume that the sum is less than $n$ where the resultant sum stored in $\log n$ locations.   We now describe a method to let one member know the sum without sending all $\log n$ bits to the node.  It is this process that fails when we try to compute a bi-symmetric functions which we will talk about later.  

All $n$ members participate in a Binary Search Tree, once as a leaf and once as an internal node of the binary search tree.  We will set the root of the tree and $\log n -1$ leftmost descendants of the tree to be the $\log n$ members with the final count bit each.  The root contains the most significant (that is $\log n$th) bit assuming the sum is less than $n$ while the significance of the bit decreases as we descend the tree in the left most path.  Each member containing the sum bit on the left-most branch of the tree will send the bit value, called control-bit, to all members on the same level of the tree in a sequential fashion.  This starts with the root, when a level finishes the message passing, the last member sends a message "0" to the leftmost member of the tree one level below the current level.  So each member sends and receives one bit.  When the last level finishes its processing, the last member sends a "0" message to the root to start the descending process.  

Each leaf node has a number which starts with $0$ and ends in $n-1$ and they appear in order from left to right.  Depending the value of the control-bit, the root sends a message "0" to the left, if the control-bit is zero or to the right child if the control-bit is one. Upon receiving a message any intermediate node will send a message "0" to the left or the right child as per the control bit it has.  At the last level, message "0" is sent the left or right child.  When a node receives the last message "0", it consults its designated number in the range $[0, n-1]$ and outputs the value of the function.  Each member sends and receives at most one bit.  Except for the last bit, the total number of bits sent and received each by a member is at most 2. Therefore local communication complexity is at most 5+3+3 = 11.
\end{proof}

For a tighter result, a more involved construction and analysis
can achieve the same complexity of $11$ even if $n$ is not a power of $2$.
We omit these details for the sake of concision.

\iffalse
When $n$ is not a power of $2$, we can easily establish a protocol with $O(1)$ local communication complexity to compute the number of $1$'s.  
Observe that there is number $x = 2^y$ such that $n \leq x \le 2n$ and asking $n-x$ players to participate "twice" where the input bit on their second participation is set to "0".  But, the messages in the protocol must explicitly distinguish between these two versions.  Therefore, an additional bit is added to every message, which at most doubles the number of bits sent/received.  Since some players participate twice, we get another factor of two resulting in upper bound of $44$ on the local complexity.  
\fi

\subsection{Solving (n, (n/10))-Counting with Local Complexity of 2}

We now turn our attention to counting with a local complexity of $2$.
We show that even with this small amount of communication, counting up to $k=\Omega(n)$
is possible.
The proof for this theorem
deploys the same general tree-based counter aggregation and subsequent
dissemination strategies introduced in  our $(n,n)$-counting solution.
We now, however,
carefully implement these strategies in such a way that the
the $(9/10)n$ nodes {\em not} counting their inputs
can collaborate with the $n/10$ nodes that {\em are} counting to reduce the number of bits
they need to send and receive  from $11$ down to $2$.

\bigskip

We begin by isolating and analyzing a key step of this efficient simulation:
how to leverage helper nodes to implement the distributed counter addition
strategy from our $(n,n)$-counting solution with a local complexity of only $2$. Recall that
the simpler implementation of this addition step in our $(n,n)$-counting solution,
in which there were no helper nodes present to reduce communication,
induced a local complexity of $5$ bits.

\begin{lemma}\label{binary-addition}
Suppose $x \ge 3$ is an integer.  Assume two sets of $x$ players with one input bit each, where each collection of $x$ bits is interpreted
as a binary integer.  
There exists a  protocol with local complexity $2$ that computes the binary sum of these two numbers and stores the result 
in a third set of $x+1$ players,
using an additional set of $2x$ players to support the computation,
leading a total of $x+x+(x+1)+2x = 5x+1$ total players involved.
 Though the overall local complexity is $2$,
the $2x$ players involved in storing the two input numbers send $1$ bit each and receive none during the computation and $x+1$ players storing and the resultant sum send at most $1$ bit each during the computation.
%$3x+1$ players involved in storing input and 
%output numbers are allowed to send at most $1$ bit in the protocol
\end{lemma}

%We will prove that the binary addition of two $x$-bit numbers can be performed by a 2-bit protocol 
%and this involves $5x+1$ players.  Initially $2x$ players contain the binary bits of these two binary numbers.  
%Input bit of other players are irrelevant.  Finally $x+1$ players, other than original $2x$ players, will contain the sum.
%We primarily assume that $x \ge 3$, since the required $5x+1$ can be reduced for smaller $x$.

\begin{proof}
%We now describe how to perform the binary addition of two 4-bit binary numbers using a protocol with local complexity $2$.  
Let  $a_x \ldots a_3a_2a_1$ be bits defining the first number a $b_x \ldots b_3b_2b_1$ be the bits defining the second.
Let $r_{x+1}r_x \ldots r_3r_2r_1$ be the resulting sum.
In the following, we use the notation $P_{a_i}$, for a given labelled bit $a_i$,
to denote the players responsible for bit $a_i$.
%Each $a_i$, $b_i$ and $r_i$ will be thought of as a player in the protocol
%and so each one is capable of sending and receiving at most two bits.  
In addition to these $3x+1$ players, we will use 
$2x$ additional players, which we label $P_{t_i}$ and $P_{c_i}$,
for each $1 \le i \le x$.

%Recall $a_i$'s and $b_i$'s contain the two binary numbers initially.  All other bits are arbitrary.  
%Even though 
%we are interested in a 2-bit protocol for addition, we impose additional restrictions on the number of bits sent by players $a_i$, $b_i$ and $r_i$  to be 
%at most $1$ so that one can apply binary addition repeatedly. 

%The computation performed by each player is symmetric function and thus the order in which the bits are received by the player is not relevant.
% CAL: didn't know what the above meant

We now describe the computation.
For each $1 \le i \le x$,  $P_{a_i}$ and $P_{b_i}$  send their bits to $P_{t_i}$. 
Upon receiving two bits, each $P_{t_i}$ computes the XOR of the two bits and sends the result to $P_{r_i}$.
This value represents a tentative sum of the relevant two bits.   
Each $P_{t_i}$ also computes the AND of these two bits, encoding the tentative carry, and sends it to $P_{c_i}$. 
For $P_{r_i}$ to compute the final sum for this bit position,
it also needs to know the relevant carry, which it can receive from $P_{c_{i-1}}$.
Similarly, for $P_{c_{i-1}}$ to know the full carry to send to $P_{r_{i}}$,
it needs to learn not just the carry bit from $P_{t_{i-1}}$,
but also any carry resulting from the sum computed by $P_{r_{i-1}}$.

%The protocol then updates the carry information to compute the sum.  
%Since both the carry bit $c_1$ and the result bit $r_1$ are correct, upon the receipt of a bit,
%$P_{c_1}$ sends the bit to $P_{r_2}$.  For ease of presentation we kept $c_1$ in our presentation.   One can eliminate this and it has no effect in the %resultant bounds.  

Let us pull together these pieces:
For $2 \leq i\leq k$, each $P_{r_i}$ and $P_{c_i}$ computes and communicates the following after receiving two bits in any order.  
$P_{r_i}$ computes XOR of the two received bits and stores it as the resultant sum bit.  
$P_{r_i}$ computes AND of the two received bits and sends it to $P_{c_i}$. 
$P_{c_i}$ computes OR of the two bits and sends it to $P_{r_{i+1}}$.  
$P_{r_{k+1}}$ stores the bit it received the most significant bit of the sum.
We can bootstrap the relevant processes in position $1$ to send the correct value on initialization
(e.g., $P_{r_1}$ has no carry bits to receive).
It is easy to verify that the resultant sum computation is correct and it meets the required communication bounds.
\end{proof}

We are now ready to describe a protocol that solves $(n, n/10)$-counting with local complexity of $2$.
This protocol will leverage the distributed adding strategy captured in the preceding lemma to 
effectively count $1$ bits among the first $n/10$ positions in an efficient manner.
The remaining $(9/10)n$ nodes will be used to implement the results, totalling, and carrying
roles needed by this addition.

\begin{theorem}%\label{n/10}
For every $n$ and $k$, such that $0 \leq k \leq n/10$,
there exists a protocol that solves $(n,k)$-counting with a local communication complexity of $2$.
\end{theorem}
\begin{proof}
Assume $m = 10n$ and so we are computing number of $1's$ in the first $n$ positions.
As in our proof of $(n,n)$-counting result, we assume for now that $n$  is a power of $2$.
(The same technique that eliminates this assumption for the $(n,n)$-counting case applies here,
but as before we omit for the sake of clarity.)
The computation follows the main idea of the proof of Theorem \ref{thm:sym}.  However, we
will use an additional $9n$ members to reduce the number of bits used in the 11-bit protocol to 2 
bits.  

We partition $m=10n$ players into 10 groups  of $n$ members each.  The goal is to count the number of $1$'s in the first group.   
The input contained in the remaining 9 groups will be ignored and the players will be used to support the
counting of $1$'s in the first group.   

The protocol is divided into six phases.  In phase 1, the second group of $n$ players
will receive information from first group in the following way. 
Partition the input bits of first group into collection of two positions at a time (say $a_i, a_{i+1}$) and perform a simple binary addition of the two bits for each n/2 collections and store the 2-bit output of 
the binary addition in the corresponding two positions (say $b_i, b_{i+1}$) in the second group. 
It is easy to see that this can be accomplished where
both players (say $a_i$ and $a_{i+1}$, $i$ is odd) in the first group send one bit to each of the two players $b_i$ and $b_{i+1}$. The players in the second group receives two bits each but each player can still send 2 bits.  

In the second phase, the third group will receive information from the second group such that every four bit of 
the third group contains the sum, in binary, of the number of 1's in the corresponding four players of the 
first group.  Note that three bits are sufficient to store the binary value between $0$ and $4$ and the 
extra position is used to reduce the number of bits sent/received to 2, as shown below.  The protocol proceeds in the following
way.  Given two binary numbers $b_2 b_1$ and $b_4 b_3$,  $c_1$ receives both bits $b_1$ and $b_3$ while 
and $c_4$ receives both bits $b_2$ and $b_4$.  $c_4$ sends  $(b_2 ~ OR ~ b_4)$ to $c_2$ and 
$(b_2 ~ AND ~ b_4)$ to $c_3$.  $c_1$ sends  $(b_1 ~ AND ~ b_3)$ to $c_2$ and stores  $(b_1 ~ OR ~ b_3)$ in $c_1$.
$c_2$ upon receiving the second bit, sends the AND of these  bits it received to $c_3$ while storing the OR of these two bits in $c_2$.
Upon receiving two bits,  $c_3$ stores the OR of these two bits in $c_3$.  
Notice that each of the $c_1$ through $c_4$ receives at most $2$ bits 
while sending at most $1$ bit.  Observe that the result of the addition of $b_2 b_1$ and $b_4 b_3$ is in $c_3,c_2,c_1$.  
Each $c_i$ can still send $1$ more bit.
The binary count in $c_3,c_2,c_1$ represents the number of $1$'s in $a_1$ through $a_4$. 
This process is repeated so that the sum of number of $1$'s in every successive group of four $a_i$'s. 
%For the ease of presentation, we explain the protocol at one end of each group.  
%However this process is repeated everywhere in each group. For instance, $c_7c_6c_5$ will 
%contain the count of number of $1$'s in $a_5$ through $a_8$ and so on.

In the third phase, we will show that $4$ players in the fifth group, namely $e_4, e_3, e_2,e_1$, 
will contain binary bits such that $e_4 e_3 e_2 e_1$ represents the count of
number of $1$'s in positions $a_1$ through $a_8$.  
In order to do so, we will perform simple binary addition of numbers $c_3,c_2,c_1$ 
and $c_7c_6c_5$ using the protocol explained in Lemma \ref{binary-addition} where the parameter $x=3$. 
We employ a total of $5x+1 = 16$ players.  
Six players from the fourth group, $d_1$ through $d_6$, four players from fifth group, 
$e_1$ through $e_4$,  and seven players from third group, $c_1, c_2, c_3, c_5,c_6$ and $c_7$ will be the 16 players that employ the protocol of Lemma \ref{binary-addition}.  It it not hard to see that the communication limitations of the Lemma \ref{binary-addition} is met.

In the fourth phase,  we set $5$ players in the seventh group, namely $g_5,g_4, g_3, g_2, g_1$, 
to contain binary bits such that $g_5g_4g_3g_2g_1$ represents the count of
number of $1$'s in positions $a_1$ through $a_{16}$. 
This is done by performing binary addition of $e_4 e_3 e_2 e_1$ and $e_{11} e_{10} e_9 e_8$.  
As before, we apply Lemma \ref{binary-addition} where the parameter $k=4$.
We use $5x+1 = 21$ players out of which 8 are carrying input bits, namely $e$'s, and 5 store the output bits, namely $g$'s.  We need additional 8 players from sixth group, namely $f_1$ through $f_8$, to perform the computation.

Starting with binary counts of $16$ $a_i$'s at a time, stored in $g$'s, 
we will perform repeated binary addition (bottom-up counting process) to compute the count of all $n$ locations.
%This is the fifth phase of the computation.
Unlike the previous four phases where we needed additional group for each addition, 
the fifth phase performs all of the binary additions starting from $5$ bits to $\log n$ bits using only one additional
group $h_1, h_2, \ldots h_n$.   
This is the eight group of $n$ players.  
As we perform repeated additions using Lemma \ref{binary-addition}, each addition involves $5x+1$ 
players out of which $3x+1$ are input/output players.
Only $2x$ players perform intermediate computations.
Since $3x+1 \ge 2x$, the number of players used to perform intermediate computations is no larger than the number of players involved in the input/output parts of the process.  The proof of Theorem \ref{thm:sym} shows that at most $n$ players are used in storing input/output part of the addition process.
Therefore, for the entire collection of additions involving 
$5$ bits to $\log n$ bits, the total number of players who perform intermediate
computations is not larger than the number of players involved in the input/output parts of the addition.  
%By Theorem \ref{thm:sym}, $g_1$ through $g_n$ contain $n$ players sufficient to store
%the intermediate sums.  
%In order to perform the binary addition specified in Lemma \ref{binary-addition}, 
%we need the eight group, namely $h_n, h_{n-1}, \ldots, h_1$.
As argued in
the proof of Theorem \ref{thm:sym}, 
if the final sum is $n$ then the most significant bit is $1$ and it can declare the output.  
Otherwise, each player of the remaining $\log n$ bits  of the count will send its bit to 
another new player within the seventh group so that these new players  have capability to
send two bits instead of only one.  
This transition is possible within seventh group ($g$'s) since the recurrence relation 
$U(k) = 3~2^{k-2} - (k+3)$ implies $U( \log n) \leq n - (\log n +3)$.
Observe that $U(k)$ is the number of locations used once in the process.  
Note that we do not use the calculation process of Theorem \ref{thm:sym} since the number of bits sent and received
by a player exceeds $2$.  
We use the availability of the locations specified in the proof of  Theorem \ref{thm:sym} and 
perform computations as per Lemma \ref{binary-addition}.
Let $g_{f(\log n)}, g_{f(\log n -1)}, \ldots, g_{f(1)}$ be the final count of number of $1$'s in binary.  
Note that $g_{f(1)}$ is the least significant bit of the count.

The sixth and the last phase will contain two groups, namely the ninth group $i_1, i_2, \ldots , i_n$ and 
the tenth group $j_1, j_2, \ldots j_n$ of players where one of $j$'s will output the correct count.  
As in the proof of Theorem \ref{thm:sym}, the players $i_1, i_2, \ldots , i_{n-1}$ will form the internal nodes of a complete binary tree
and the players $j_1, j_2, \ldots, j_n$ will be the leaves.  
The root, we label level 1, of the binary tree will get the most significant bit, namely $g_{f(\log n)}$.  
The nodes in level $k$ will get $g_{f(\log n -(k-1))}$.   
The distribution of this bit starts with the least significant bit first and end with the most significant bit to the root.  
This process does not proceed in parallel but in a sequential process using token passing.  
This requires $g_{f(\log n)}, g_{f(\log n -1)}, \ldots, g_{f(2)}$ to receive a bit to start the seeding of the next level. 
This ensures that the first bit received by the nodes of the tree is the control bit.  
Once the root gets its control bit, it follows the binary search process described in proof of Therorem ~\ref{thm:sym}.
Only one leaf will receive a bit and it declares the correct count based on its position in the tree.
\end{proof}

\subsection{Solving Sorting and Searching with Constant Local Complexity}

The ability to count $1$ bits in the input with constant local complexity
enables the solution of other natural problems with this same low complexity.
We highlight two such problems here:

\bigskip

{\bf Sorting:}
The first problem is {\em sorting}.
In the context of binary inputs distributed among $n$ nodes, $P_1,P_2,\ldots,P_n$,
sorting reduces to gathering all the $1$ bits together.
Formally, if $\#_1(a_1\ldots a_n) = k$,
then to solve sorting for input $a_1\ldots a_n$,
the nodes in $\{P_1,\ldots,P_k\}$ should output $1$,
and the nodes in $\{P_{k+1},\ldots,P_n\}$ should output $0$.

The counting solution described in analyzed in Theorem~\ref{thm:counting}
has the nice property that not only does a node output $k$,
but the unique node that does so is $P_k$.
To extend this solution to sorting, therefore,
it is sufficient for $P_k$ to disseminate a bit down the line from $P_{k-1}$ to $P_1$,
letting these preceding nodes know that they should also output $1$.
This increases the local complexity by a single bit from $11$ to $12$.
Formally:

\begin{theorem}
The sorting of 1-bit inputs can be solved with local complexity $12$.
\end{theorem}

\bigskip 

{\bf Searching:}
Another natural problem is {\em searching.}
In particular,
for a given binary input assignment $a_1\ldots a_n$,
we say that $P_i$ has the $k$th ``$1$'',
if: (1) $a_i = 1$, and (2) $\#_1(a_1\ldots a_i) = k$.
We can therefore define a search problem, parameterized with  $k$,
such that the goal is to output the id of the node with the $k$th one.
Building on our counting strategy,
we can also solve this problem with constant local complexity:

\begin{theorem}\label{kth}
For every network size $n = 2^i$ and search location $k \le n$, searching for the 
$k$th $1$ can be solved with
constant local complexity. 
\end{theorem}
\begin{proof}
First, run the protocol for  counting the number of ones in the  input where the result is in the binary form.  This is the first tree of the Theorem~\ref{thm:sym}.

Now store $k$ in binary in $\log n +1$ locations.  If $k=n$ then check the count is equal to $n$ and report the outcome.  

From now on we assume that $k< n$ and if there is $k$ ones it will be between $0$ through $n-1$.  This occupies $\log n$ bits in binary form.  
Let the count of $1$s to be $c(1,n)$.  Compare $k$ with $c(1,n)$.  This comparison can be done in $O(1)$ local communication complexity by starting with comparing most significant bit  first and find out if $k > c(1,n)$.  If $k > c(1,n)$ then there is no such $k$ exists.
Store $k$ in exactly the same place where $c(1,n)$ is stored.  
Recall that we have the counts $c(1,n/2)$ and $c(n/2+1,n)$ whose sum led to $c(1,n)$ still stored in appropriate places.  We compare $k$ with $c(1,n/2)$ in $O(1)$ local communication complexity to find out if $k > c(1,n/2)$.  

If the answer is yes, then perform subtraction $k-c(1,n/2)$.  The subtraction is very similar to the addition we performed in the Theorem~\ref{thm:sym}.  Store the result in exactly the same place where the bits of count $c(n/2+1,n)$ are stored.
Now you can compare this new $k$ with $c(n/2+1, 3n/4)$. 

But if the answer is no, then copy $k$ into where bits of $c(1,n/2)$ are stored and compare this with the count $c(1,n/4)$.

This process continues until we hit a final group of $16$ bits.  One can then easily find the original $k$th one in constant local communication complexity. 
\end{proof}

%% file: arithmetic.tex
\section{Modular Arithmetic with Constant Local Complexity}
\label{sec:arithmetic}
%A function $f$ is called 2-symmetric (or bi-symmetric) if the binary input %bits can be split into two groups such that the function value does not %change when we permute inputs within any of the two groups. 

%In the previous section, we proved that every symmetric function
%can be computed with local communication complexity of only 11.
In this section we turn our attention to 2-symmetric (also known as bi-symmetric) functions, in which the input bits can be partitioned into two sets, and the output of the function depends only
on the total count of $1$ bits in each set.

We focus in particular on balanced 2-symmetric functions, 
of the form $f:\{0,1\}^{2n} \rightarrow \mathbb{N}$,
where the partitions evenly divide the bits into two sets of size $n$.
One can  therefore interpret a function 
$f(a_1, a_2, \ldots a_n, a_{n+1}, \ldots a_{2n})$ of this 
type as calculating 
$g_f( \#_1(a_1...a_n), \#(a_{n+1}...a_{2n}))$,
for a function of the form $g_f : [n] \times [n] \rightarrow \mathbb{N}$.

%\iffalse
% I don't understand the below, so I took it out for now. We can add it back if we can clarify what is being said.
We turn our attention balanced 2-symmetric functions in part because they are the natural next
class to consider after we established in the previous section that symmetric functions can be solved
with constant local complexity.
We emphasize that the local complexity jump from symmetric to 2-symmetric is non-trivial.
A straightforward counting argument establishes that there must {\em exist} 2-symmetric functions
with a local complexity in $\omega(1)$.
Identifying a specific function with this larger complexity would resolve a major open problem in circuit complexity.
This follows due to ability of any linear-sized circuit to be simulated in our setting by a protocol
with constant local complexity (e.g., see the discussion in~\cite{boyle2018bottleneck}). 
A function that cannot be solved with constant local complexity is a function that cannot be implemented
by a linear-sized circuit.

Specifically, we begin by studying many standard modular arithmetic functions on two operands.
Perhaps not surprisingly, given the connection between local and circuit complexity,
we identify solutions to all functions considered that require only constant local complexity.
We then build on these solutions to show that even the more complex GCD function can be implemented
with constant complexity. These results underscore the surprisingly power of distributed function computation
with a very small number of bits sent and received at any one node.

\subsection{Standard Arithmetic Functions with Constant Local Complexity}

%A major open problem in circuit complexity is to find an {\em explicit} function that cannot be computed by any linear-sized circuit.  Since linear-sized circuits can be simulated by constant local communication protocol, the proof of which is not presented in this paper, one way to approach this problem is to find an explicit function that cannot be computed in constant local communication complexity.  Counting argument shows that there are many such balanced 2-symmetric functions that cannot be computed in constant local communication complexity and we hope to find an explicit one.  An obvious first step is to consider balanced 2-symmetric functions that can be represented as basic arithmetic operations on the two counts. 

%\fi

We begin by studying standard arithmetic functions, including basic mathematical and comparison operations,
and their composition.
In all cases, we prove that constant local complexity is sufficient.
A reasonable starting place for designing a protocol to compute 
a given balanced 2-symmetric function $g_f$ is to first run two instances of
our $(n,n)$-counting solutions in parallel on $a_1...a_n$ and $a_{n+1}...a_{2n}$.
This allows some player $i$ in the first partition to learn $x=\#_1(a_1...a_n)$ and 
some player $j$ in the second to learn $y=\#_1(a_{n+1}...a_{2n})$.
If $i$ shared $x$ with $j$ (or vice versa) then the function can be computed.
These values, however, might be much too large to directly share while preserving constant local complexity.
Therefore, more complexity is required for each partition to learn the necessary information
about the other partition's count.

Formally, we prove the following is possible with constant local complexity: 

\iffalse
Here we explore a collection of such strategies that enable the computation of many
2-symmetric functions that implement standard modular arithmetic operations,
including comparison, addition, subtraction, multiplication, division, and exponentiation.
These results significantly expand the contours of the complexity class containing functions computable with
constant local complexity,
and point toward efficient approaches for solving more involved problems with algebraic representations.
Formally:
\fi

\iffalse

The inputs bits provided to our $n$ players can be interpreted as function inputs in many different ways.
We focus here
on a natural interpretation arising from 2-symmetric functions.  One of our goal is to construct a 2-symmetric function without constant local complexity.  Counting argument ~\ref{} shows that there are 2-symmetric functions with local complexity $\omega(1)$. 
Recall that 2-symmetric functions partition input bits into two partitions.
For the ease of presentation, we focus on the case where the two partitions of the input are equal in size and so the 2-symmetric function $f:\{0,1\}^{2n} \rightarrow \mathbb{N}$ can be viewed as a function $g_f : [n] \times [n] \rightarrow \mathbb{N}$.  

We will show that if $g_f$ is a simple modular arithmetic function such as addition, subtraction, multiplication and division then the function can be computed with constant local complexity.

\fi

\begin{theorem}
Suppose $n+1$ is a prime number.
A balanced 2-symmetric function $f:\{0,1\}^{2n} \rightarrow \mathbb{N}$ can be computed with constant local complexity if the corresponding function $g_f : [n] \times [n] \rightarrow \mathbb{N}$ is:\\
\iffalse
\begin{itemize}
    \item a function of only one of the two variables with range $\{0,1, \ldots , n\}$,
    \item comparison operator ($>, \geq, = ,\leq, <)$ on two variables with output $0$ or $1$,
    \item a modular (mod (n+1)) addition, subtraction, multiplication or division, exponentiation or
    \item a constant applications/composition of functions from first three items in this list. 
\end{itemize}
\fi
$~~~~~$ 1. a function of only one of the two variables with range $\{0,1, \ldots , n\}$,\\
$~~~~~$ 2. comparison operator ($>, \geq, = ,\leq, <)$ on two variables with output $0$ or $1$,\\
$~~~~~$ 3. a modular (mod (n+1)) addition, subtraction, multiplication or division, exponentiation or \\
$~~~~~$ 4. a constant applications/composition of functions from first three items in this list.

\end{theorem}

\begin{proof}
We start with some basic assumptions that  $n+1$ is prime, and all arithmetic operations are modular $n+1$.  There is a natural linear order among numbers in $\{0,1, \ldots , n\}$ and so the comparison operators are well defined.

Given a 2-symmetric function $f:\{0,1\}^{2n} \rightarrow \{0,1, \ldots , n\}$, 
we will describe a protocol which notify the player $P_i$ to be responsible for the output if the output of $f$ is $i$.  

First we apply the counting protocol for $(n,n)$-Counting to find the number of  $1$'s in each of the two partitions.  
Let the two counts be $i$ and $j$ and players $P_i$ and $P_j$ know the count each.  But, player $P_i$ does not know $j$ and player $P_j$ does not know $i$.  However, the player knows if it has the first or second count.  

The protocol to compute the target function will consists of multiple phases where each phase is well defined.  The number of phases will be a constant.  Therefore any message from player $P_a$ to $P_b$ will consists of a phase number followed by the message.  For the ease of presentation, we will omit phase number information from the message exchanged between players.  Since the number of phases is a constant the asymptotic bound on the local complexity does not change if we omit phase number in the calculation.

We start with a protocol for computing a function that depends on only one of two variables, say $i$.    Given the function and the input $i$, $P_i$ computes the function by calculating the output $x \in \{ 0,1, \ldots , n\}$. $P_i$ sends a message to $P_x$, thus completing the computation.

We now establish protocol to determine if $i=j$, $i > j$ or $i < j$.
It is trivial to see if $i=j$, since player $i$ has both counts.  If $i \neq j$, player $i$ initiates a protocol with two phases.  Suppose $i >0$.  The first phase starts with $P_i$ sending $0$ to $P_{i-1}$.   Each player $P_k$, upon receiving a bit $0$ sends bit $0$ to the player $P_{k-1}$ until it reaches either player $j$ or $0$.  We call this a $0$-down sweep.  
We now consider the case where $P_j$ receives $0$.  $P_j$ will initiate the second phase by sending bit $1$ to $P_{j+1}$.  Each player, upon receiving $1$ sends $1$ to $j+1$ which will ultimately reach $i$, $1$-up sweep. 
Upon receiving $1$, $P_i$ knows that $i > j$.  

On the other hand, suppose $P_0$ receives $0$ and $j \neq 0$.  $P_0$ initiates $0$-up sweep which will terminate at $P_i$.  Upon receiving $1$, $P_i$ knows that $i < j$.

$P_i$ will inform $P_0$ through $P_n$ the result by sending
a bit $b$ (= $0$ if $i<j$ and $1$ otherwise) to $P_0$.  
Upon receiving $b$, $P_0$ will start a $b$-up sweep which ends with $P_n$.

We now consider the modular arithmetic.  The important operation is subtraction.   We now present a protocol to compute $(i-j)$.    $P_i$ knows that it has the first parameter and $P_j$ knows that it has the second parameter.
If $i \ge j$ then $x= i-j$ else $x = (n+1)-(j-i)$.  Note that $x \equiv (i-j)~ mod~ (n+1)$.  At the end of the protocol, $P_x$ knows it has $x \equiv (i-j)~ mod~ (n+1)$.

We run the protocol to determine if $i \ge j$ or not.
If $i=j$ the $P_i$, informs $P_0$ the answer, thus ending the computation.  
Now, suppose $i > j$.  Suppose $j >0$, a new phase starts at $P_0$ which initiates a $0$-up sweep which terminates at $P_j$.  During this sweep a special $0$ is stored at each player $P_0$ through $P_{j-1}$.  Upon receiving a bit, either from $P_{j-1}$ or from $P_i$ (happens when $j=0$), $P_j$ initiates $1$-up sweep which terminates at $P_i$. During this sweep, the players  $P_j$ through $P_{i-1}$ store a special $1$ each.  Now $P_i$ starts a $0$-up sweep which terminates at $P_n$.  During this sweep, the players $P_i$ through stores a special $0$ each.  Now, run the protocol to compute the number of special $1$'s among players $P_0$ through $P_n$.  The result is $(i-j)~ mod~ (n+1)$.

We now consider the case $i < j$.  We run three up sweeps.  This time we store special $1$'s at players $P_0$ through $P_{i-1}$ and $P_j$ through $P_n$ while storing special $0$'s at $P_i$ through $P_{j-1}$.   Counting the number of special $1$'s among players $P_0$ through $P_n$ results in $(i-j)~ mod~ (n+1)$.  This completes the computation of subtraction.

We now consider the computation of addition $i+j$.  
$P_j$ computes $-j$ and sends a message to $P_{n+1-j}$.
Observe $i+j \equiv i - (n+1-j)  ~mod ~n+1$.

Given $i$, $P_i$ computes $x \equiv i^2 mod (n+1)$ and sends a message to $P_x$, thus computing $i^2$.  
Simple algebra shows that multiplication can be computed by following the computations of the equation $4ij \equiv (i+j)^2 - (i-j)^2 ~mod~(n+1)$.  If $P_y$ knows $4ij$ then it can compute $z \equiv ij ~mod~(n+1)$ and 
sends a message to $P_z$.  All of the above operations can be computed even when $n+1$ is not a prime number.

When $n+1$ is a prime, inverse exists for every $1 \le i \le n$.  Therefore, division can be computed by observing $i/j \equiv i j^{-1} ~mod~ (n+1)$.  The protocol for division is straightforward and works if $n+1$ is a prime. 

Finally, we present a protocol to compute $i^j ~ mod ~ (n+1)$.  There is primitive root $g$ modulo $(n+1)$ such that for every $1 \le a \le n$ there is an integer $a'$ ($0 \le a' < n$) such that $a \equiv g^{a'} ~mod~ (n+1)$.   All players are aware of the generator and so $P_i$ finds $a$ such that $i \equiv g^a ~mod~ (n+1)$ and sends message to $P_a$.  
Now $P_a$ and $P_j$ runs the multiplication (modulo $n$) protocol to compute $b \equiv aj ~mod~ n$ which result in a message to $P_b$. 
Finally $P_b$ computes the result of the exponentiation $c \equiv g^b ~mod~(n+1)$ and sends a message to $P_c$.

If the output of an operation is $i$ the $P_i$ knows it. Therefore the local complexity of constant number of repeated applications such operations is also another constant.   
\end{proof}

\subsection{GCD with Constant Local Complexity}

%Having showing 
While basic arithmetic functions can be computed with constant local complexity,
we turn to a seemingly more complex challenge: computing the greatest common divisor (GCD) of
the two set sizes. Standard GCD algorithms require the maintenance of a super-constant size value
over a super-constant number of rounds. Implementing such strategies therefore seem to necessarily
require a super-constant local complexity. Here we prove, however, that this is not fundamental
by describing and analyzing a protocol that solves the GCD problem with constant local complexity.
The protocol is more involved than any of the strategies we have previously discussed.  
It requires the nodes to first efficiently divide the two operands into their prime factors
and use a tree-based strategy to multiply these factors in such a way that we can extract
the common divisor. 

Formally:

%
% I couldn't follow this. 
%
\iffalse
At a high-level,
the strategy has you
first break the first value $x$ into its prime (i.e., $p_i^j$) factors.  Similarly break $y$ into its prime factors.  Once you have the prime factors of both operands,
you can calculate, with care, which factors are common between the two.  This would give you powers of a prime individually.  The problem is to multiply the powers of prime together.  In order to do so we create a tree of prime factors where numbers $1$ through $n$ are the leaf nodes.
We start with the highest prime and go down to the next prime and go all the way to $2$.  As we descend the prime factor tree, we first consider the highest prime factor first.  
While a power of prime $p_i^j$ can appear in many different numbers, we can make sure that $p_i^j$ will link with all of them horizontally in a chain-like fashion that the correct one is selected to move on to the next prime factor.  
We now formalize these ideas and present a formal proof below.
\fi

\begin{theorem}
\label{arithmetic:gcd}
For $1 \leq x,y \leq n$, the local complexity of computing $GCD(x,y)$ is $O(1)$.
\end{theorem}

\begin{proof}
Let $x$ and $y$ be the two inputs such that $1 \leq x ,y \leq n$.  
Let the node be numbered $1$ through $n$ such that the node $x$ knows $x$ and node $y$ knows $y$.  If any of the two is a $1$ then the answer is $1$.  So we may assume now on that neither of the two is a $1$.  

Let $p_1 = 2$, $p_2 = 3$, $p_3 = 5$ and $p_i$ be the ith prime number.
Let $x = p_{r(1)}^{m_1} p_{r(2)}^{m_2} ... p_{r(k)}^{m_k}$ and 
$y = p_{s(1)}^{n_1} p_{s(2)}^{n_2} ... p_{s(l)}^{n_l}$ where 
$r(1) < r(2) < \ldots < r(n)$ and $s(1) < s(2) < \ldots < s(n)$.  The node $x$ will send a message "0" to two nodes.  
It will start to send a message to  the node $p_{r(1)}^{m_1}$ and the node $p_{r(2)}^{m_2} ... p_{r(k)}^{m_k}$.  
Now the node $p_{r(1)}^{m_1}$ will send a message "0" to the node $p_{r(1)}^{m_1-1}$.  
This message will be propagated until it reaches $p_{r(1)}$ with a message "0".  
Now the node $p_{r(2)}^{m_2} ... p_{r(k)}^{m_k}$ will send a message "0" to the node at $p_{r(2)}^{m_2}$ 
and to the node at $p_{r(3)}^{m_3} ... p_{r(k)}^{m_k}$.
This way one is able to distribute $x$ to $p_{r(1)}^{m_1}$ to $p_{r(1)}$, $p_{r(2)}^{m_2}$ to $p_{r(2)}$ and so on to $p_{r(k)}^{m_k}$ to $p_{r(k)}$.  Every node participates at most once in this process and the message is "0".
Similarly, the $y$ distributes its message of "1" to $p_{s(1)}^{n_1}$ to $p_{s(1)}$, $p_{s(2)}^{n_2}$ to $p_{s(2)}$ and so on to $p_{s(l)}^{n_l}$ to $p_{s(l)}$.  Every node participates at most once in this process of sending $y$ and the message is "1". 

Now we start sending "00" from the node $n$ to $1$. It is important that a "00" message does not over take a "0" or a "1" message.  This make sure that both messages "0" and "1" are reached their destinations.  
After a "00" message reaches the node $1$, the next stage begins. 

Let $p_{s(i)}$ be a prime number.  If there are two messages, namely a "0" and "1" reaches a node $p_{s(i)}^m$ then that node sends a message "0" to a node $p_{s(i)}^{m+1}$, if it exist, to find it has both messages "0" and "1". If it has both messages, then it sends the message "0" to $p_{s(i)}^{m+2}$ and otherwise, it sends a message back "1" to $p_{s(i)}^{m}$.  If you have received a message "1" then you are the lucky one to get both and you get to participate in the next phase.  On the other hand, if you are the last one to get both messages a "0" and a "1" then you are get to participate in the next phase.   This process takes place everywhere there are two messages, a "0" and a "1".  
Now, unlike the previous case where we sent a "00" top down, this time we send a bottom up from $1$ to $n$ a message "00".  
Like before, the message "00" does not have a preference over a message "0" or a message "1".  This "00" message makes sure that every message  reaches the destination.  At the end of the "00" message, we would have made sure that the message $p_{s(1)}^{o_1}$, $p_{s(2)}^{o_2}$, ... , $p_{s(j)}^{o_j}$ has been set where these are prime powers and $s{1} < s{2} < \ldots < s{j}$.
It is important to note that the correct answer is the product  $p_{s(1)}^{o_1} p_{s(2)}^{o_2} ... p_{s(j)}^{o^j}$.  But we have individual numbers $p_{s(1)}^{o_1}$, $p_{s(2)}^{o_2}$ ... $p_{s(j)}^{o^j}$.

\begin{figure}
\includegraphics[scale = 0.55]{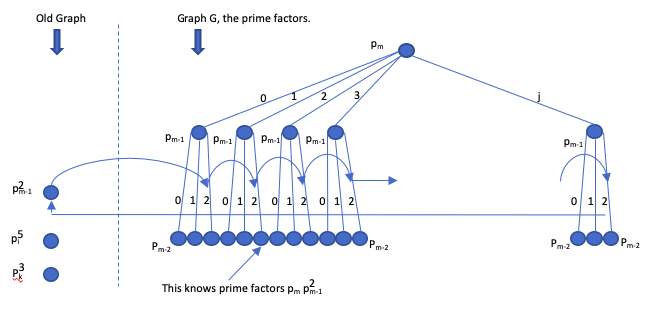}
\caption{Example of the tree structure (with horizontal chains) used to efficiently
calculate the product of prime factors of the operands.}
\end{figure}

Unfortunately we cannot apply a simple multiplication since there are too many of them.  Never the less, we can perform the operations and we show how to do so here.   
%Let $r_1 = 2, r_2 = 3, r_3= 5, ... , r_m$ be the consecutive prime number such that $p_m$ is the maximum prime number less than equal to $n$.
We form the edge label graph $G$. This graph $G$ has at most $2n$ nodes in it and so can be simulated by the network of size $n$.  
Let $m$ be set such that $p_m \leq n$ and $p_{m+1} > n$.
We form an edge labeled tree such that the root is a prime number $p_m$ and it has at least 2 branches, namely a branch $0$ and a branch $1$.  
In fact, a generic $p_j$ and $p_{j-1}$ has as many branches possible until we exceed $n$. The binary tree has exactly $i+1$ nodes where $p_j^i \leq n <  p_j^{i+1}$. Label the edges by $0$, $1$, and so on until we reach the last edge $i$.  

At level 1, we have vertices all labeled $p_{m-1}$. 
These are connected to the root which is $p_m$ and it is at level 0.  
Suppose an edge between $p_m$ and $p_{m-1}$ is labeled $j$. 
Following this, the node $p_{m-1}$ splits off at most $k$ times where 
$k$ is given by $p_m^j p_{m-1}^{k-1} \leq n < p_m^j p_{m-1}^k$.
The labels of these $k$ edges are  $0$, $1$, 
$\ldots$ $k-1$. This happens to all the nodes at level 1.

At level 2, we have vertices all labeled $p_{m-2}$.
These are connected to a vertex $p_{m-1}$ in level 1.
Suppose an edge between $p_m$ and $p_{m-1}$ is labeled $j$
and subsequently an edge between $p_{m-1}$ to $p_{m-2}$ is labeled 
$k$.  
Following this, the node $p_{m-2}$ splits off at most $l$ times where
$l$ is given by $p_m^j p_{m-1}^{k-1} p_{m-2}^{l-1} \leq n < p_m^j p_{m-1}^{k-1} p_{m-2}^l$.  
The labels of these edges are 
$0$,
$1$, $\ldots$
$l-1$.

We go all the way to $p_0$ where $p_0 = 1$.  
We have individual nodes from the original graph with label $p_{s(1)}^{o_1}$, $p_{s(2)}^{o_2}$ ... $p_{s(j)}^{o^j}$.  These are the ones to initiate the chain that link the edges horizontally.
In the edge labeled graph $G$, there is only one root, namely $p_m$.  There are many branches from the root, all of them lead to different nodes but they are named $p_{m-1}$.  For instance, a node named $p_{m-1}$ is linked to many nodes each named $p_{m-2}$.  This goes on until we hit $p_0$.  For each $k$ where 
$1 \le k \le j$, we have a chain that connects edges horizontally between edges labeled $o_k$ and lie between the nodes $p_{s(k)}$  and $p_{s(k)-1}$. 
This chaining process is linear and goes from one edge to another and to another until we hit the last edge connects back to $p_{s(k)}^{o_k}$ in the original graph.  This completes the process of building the tree.

Now notice that every internal node of $G$ has at least 2 branches. 
This is so since if we have one branch then the corresponding prime does not exists and so we move on to the next one.  
Therefore the resultant tree has at most $n-1$ internal nodes and $n$ leafs.  Therefore each node participate in at most $O(1)$ times. 

Now we are back to the protocol. 
We have individual numbers $p_{s(1)}^{o_1}$, $p_{s(2)}^{o_2}$, $\ldots  p_{s(j)}^{o_j}$.
We assume that  $s(1) < s(2) < \ldots < s(j)$. 
These individual nodes of the original graphs and we need to perform the multiplication of them. 

The first one is $p_{s(j)}^{o_j}$.   
But many edges have label that $o_j$ between nodes $p_{s(j)}$ and $p_{s(j)-1}$ in $G$ and they are chained together in the following way.  Send a message "1" from the node $p_{s(j)}^{o_j}$ to the graph $G$ marking all of the edges marked $o_j$  and between $p_{s(j)}$ and $p_{s(j)-1}$ with a "1". This is done in a cyclic fashion where the original sender sends a bit which is then forwarded by the edge one after another until the last edge sends it back to the sender.  
This take one message and hence a constant. 
Similarly, chain all edges in $G$ that corresponds to $p_{s(j-1)}^{o_{j-1}}$.  Repeat the process until finish chaining $p_{s(1)}^{o_1}$.
Since the edges are disjoint, we send one message to accomplish this task.  Finally, we mark every other edges a "0" in a cyclic fashion so that every one sends $O(1)$ messages.  This completes the chaining process.
    
A "11" bit has been sent to the node with the largest prime $p_m$ in the graph $G$ by the node $p_{s(j)}^{o_j}$ of the original graph. Now, the node $p_m$ of $G$ cycles through all the edges of $p_m$  to check if any of them received a special $1$ in the chaining process.   If it finds a bit "1" then it takes the edge and other wise it ends with first choice, that is $p_m^0$.  Note that this takes only one message from each node since they take turn in sending message to the next node.  After having found no $1$ message, the protocol takes the choice $0$.   
It then goes to the second prime, namely $p_{m-1}$ and search for a $1$ in the chaining process of $G$.  
After not finding it, it takes the choice of $0$ again.  
This way it goes to finally $p_{s_j}$ and finds the option $o_j$ with a "1" and takes it.  It means that the number is any where between $p_{s(j)}^{o_j}$ and the number obtain by multiplying this number by any prime numbers strictly smaller than $p_{s(j)}$.  From now on, your path lies under this option.  But the process continues and it looks at $p_{s(j)-1}$ and so on until it looks at $p_{s(j-1)}$.  In the mean time, all of $p_{s(j)-1}$, $p_{s(j)-2}$ etc. will be set to a zero until it hits $p_{s(j-1)}$. In this case, all of edges labeled $a$ be zero except $o_{j-1}$.  It will then add $p_{s(j-1)}^{o_{j-1}}$ to
%$p_{s(j)}^{o_{s(j)}}$ thus making $p_{s(j)}^{o_{s(j)}} p_{s(j-1)}^{o_{s(j-1)}}$.   
$p_{s(j)}^{o_j}$ thus making $p_{s(j)}^{o_j} p_{s(j-1)}^{o_{(j-1)}}$.
But this means that we have a range of values starting from this number to anything that is a product of this number and the number obtained by multiplying this number with any smaller prime numbers, that is prime numbers less than $p_{s(j-1)}$.  From now on, you are under these two options. 
This process continues until it completes the $p_1$ there by completing the $p_{s(j)}^{o_j} p_{s(j-1)}^{o_{j-1}} \ldots p_{s(1)}^{o_1}$.
The result is the product of these primes which is what you need.
\end{proof}